\documentclass[sigconf]{acmart}
\makeatletter
\def\@ACM@checkaffil{
    \if@ACM@instpresent\else
    \ClassWarningNoLine{\@classname}{No institution present for an affiliation}%
    \fi
    \if@ACM@citypresent\else
    \ClassWarningNoLine{\@classname}{No city present for an affiliation}%
    \fi
    \if@ACM@countrypresent\else
        \ClassWarningNoLine{\@classname}{No country present for an affiliation}%
    \fi
}
\makeatother
\usepackage{caption}
\usepackage{amsthm}
\usepackage{enumitem}
\usepackage{subcaption}
\usepackage{tikz}
\usepackage{pgfplots} 
\pgfplotsset{compat=1.16} 
\usepackage{balance}
\usepackage{bigstrut,array,multirow,tabularx}
\usepackage{caption}
\usepackage{dsfont}
\usepackage{mathtools}

\usepackage{algorithm}
\usepackage{algorithmic}
\usepackage{makecell} 
\AtBeginDocument{%
  \providecommand\BibTeX{{%
    \normalfont B\kern-0.5em{\scshape i\kern-0.25em b}\kern-0.8em\TeX}}}

\copyrightyear{2023}
\acmYear{2023}
\setcopyright{acmlicensed}\acmConference[KDD '23]{Proceedings of the 29th
ACM SIGKDD Conference on Knowledge Discovery and Data Mining}{August
6--10, 2023}{Long Beach, CA, USA}
\acmBooktitle{Proceedings of the 29th ACM SIGKDD Conference on Knowledge
Discovery and Data Mining (KDD '23), August 6--10, 2023, Long Beach, CA,
USA}
\acmPrice{15.00}
\acmDOI{10.1145/3580305.3599292}
\acmISBN{979-8-4007-0103-0/23/08}

\begin{document}

\title{Criteria Tell You More than Ratings: \\
Criteria Preference-Aware Light Graph Convolution \\for Effective Multi-Criteria Recommendation}


\author{Jin-Duk Park}
\affiliation{%
  \institution{Yonsei University}
}
\email{jindeok6@yonsei.ac.kr}

\author{Siqing Li}
\affiliation{%
  \institution{The University of New South Wales Australia}
}
\email{siqing.li@unsw.edu.au}

\author{Xin Cao}
\affiliation{%
  \institution{The University of New South Wales Australia}
}
\email{xin.cao@unsw.edu.au}

\author{Won-Yong Shin}
\authornote{Corresponding author}
\affiliation{%
  \institution{Yonsei University}
}
\email{wy.shin@yonsei.ac.kr}




\renewcommand{\shortauthors}{Jin-Duk Park et al.}

\begin{abstract}
The multi-criteria (MC) recommender system, which leverages MC rating information in a wide range of e-commerce areas, is ubiquitous nowadays. Surprisingly, although graph neural networks (GNNs) have been widely applied to develop various recommender systems due to GNN's high expressive capability in learning graph representations, it has been still unexplored how to design MC recommender systems with GNNs. In light of this, we make the first attempt towards designing a GNN-aided MC recommender system. Specifically, rather than straightforwardly adopting existing GNN-based recommendation methods, we devise a novel \textit{criteria preference-aware} light graph convolution (\textsf{CPA-LGC}) method, which is capable of precisely capturing the criteria preference of users as well as the collaborative signal in complex high-order connectivities. To this end, we first construct an \textit{MC expansion graph} that transforms user--item MC ratings into an expanded bipartite graph to potentially learn from the collaborative signal in MC ratings. Next, to strengthen the capability of criteria preference awareness, \textsf{CPA-LGC} incorporates newly characterized embeddings, including {\em user-specific criteria-preference embeddings} and {\em item-specific criterion embeddings}, into our graph convolution model. Through comprehensive evaluations using four real-world datasets, we demonstrate (a) the superiority over benchmark MC recommendation methods and benchmark recommendation methods using GNNs with tremendous gains, (b) the effectiveness of core components in \textsf{CPA-LGC}, and (c) the computational efficiency.
\end{abstract}





\maketitle

\section{Introduction}
\label{section 1}

\begin{figure}[t]
        \centering
        \begin{subfigure}[c]{0.49\columnwidth}
                \includegraphics[width=0.99\columnwidth]{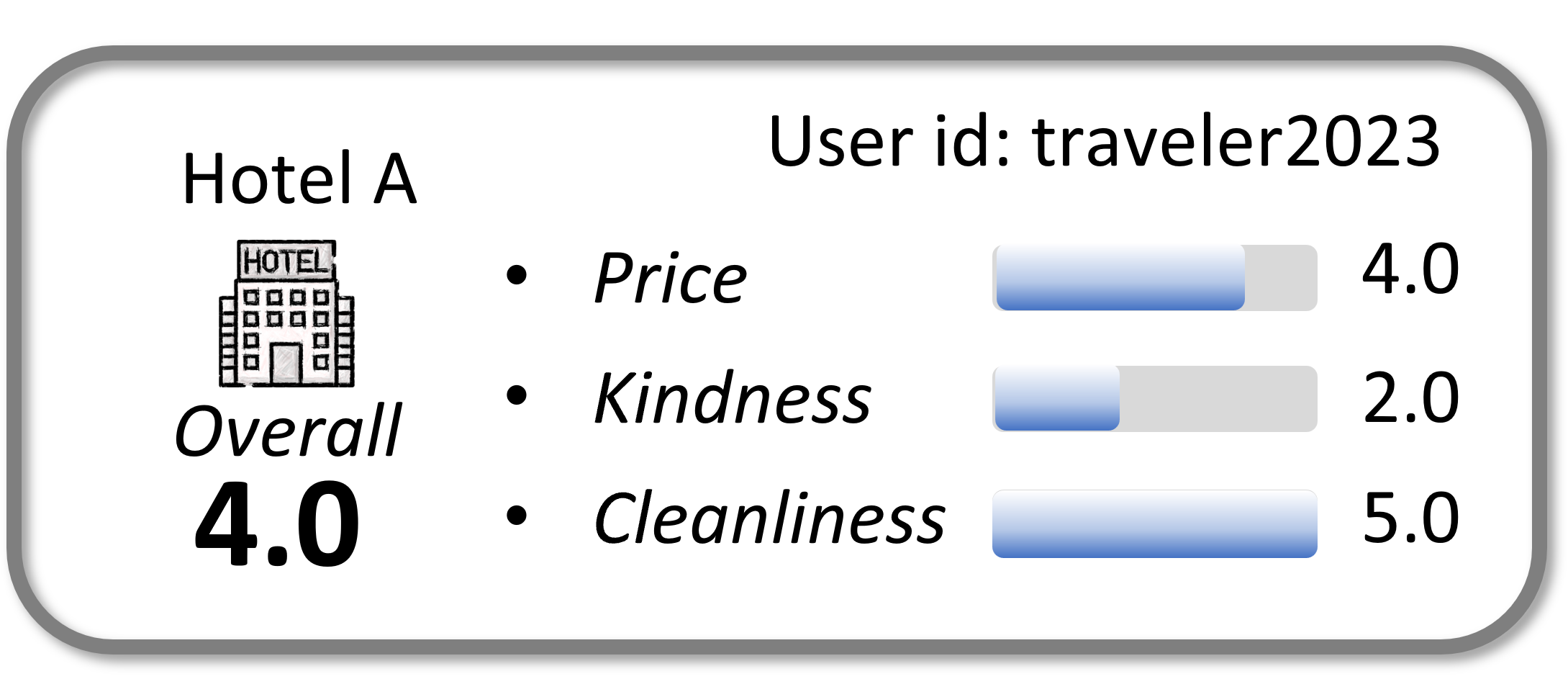}
                \caption{}
                \label{fig:intro_1a}
        \end{subfigure}        
        \begin{subfigure}[c]{0.49\columnwidth}
                \includegraphics[width=0.99\columnwidth]{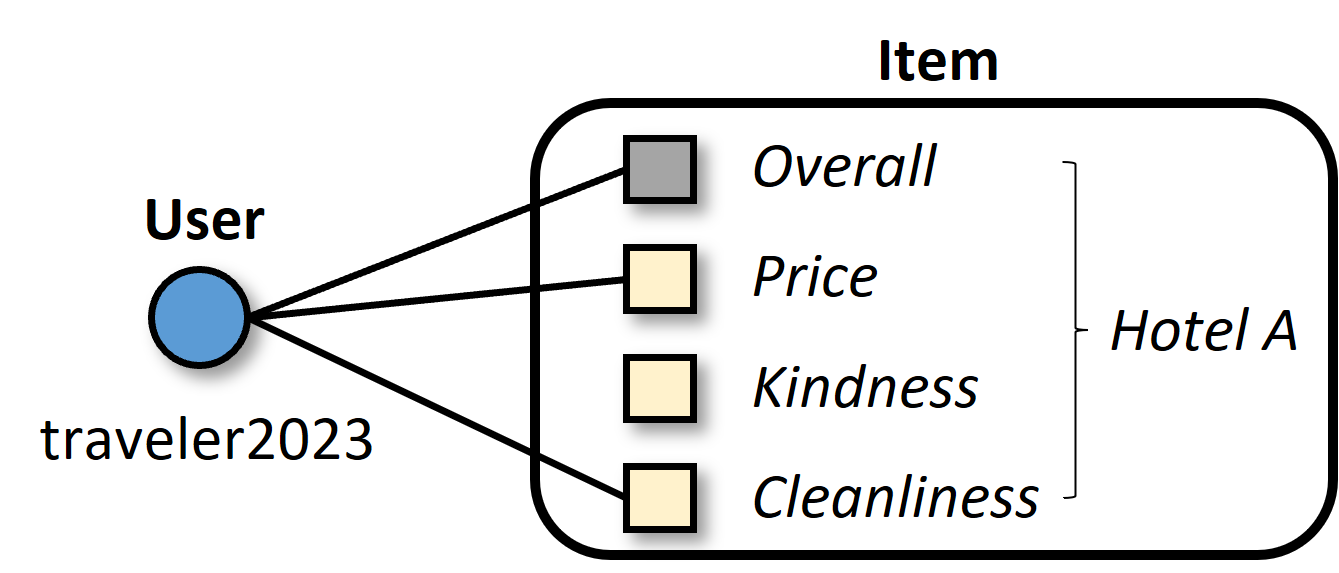}
                \caption{}
                \label{fig:intro_1b}
        \end{subfigure}

        \caption{ An illustration showing (a) four criteria ratings in a hotel domain, and (b) the corresponding MC expansion graph. }
        \label{dist_plot} 
\end{figure}

The so-called {\it multi-criteria (MC) recommender system} has been increasingly valuable for improving the recommendation accuracy based on enriched information in terms of {\it criteria} ratings of each item in a wide range of service areas such as restaurants, hotels, movies, music, {\it etc.} \cite{wang2011latent, shambour2021deep,jannach2014leveraging,nassar2020novel, li2019latent,mcauley2012learning}. MC recommender systems usually have a better ability in recommending relevant items to a user, compared to single rating systems, while better representing predilections of users \cite{fan2021predicting, nassar2020novel, shambour2021deep, zheng2019utility, li2019latent,jannach2014leveraging}. For instance, as illustrated in Figure \ref{fig:intro_1a}, a user can provide four criteria ratings (including {\it overall} ratings) in a hotel domain with the criterion being price, kindness, and cleanliness.

On one hand, collaborative filtering (CF), which exploits similar patterns learned from user--item historical interactions to recommend the most relevant items to a user, has emerged as one of the most common approaches to building recommender systems \cite{wang2019neural,he2020lightgcn,wang2020disentangled,mao2021ultragcn}. Due to its effectiveness and ease of implementation, most of existing MC recommender systems have been developed using various CF-based techniques such as matrix factorization \cite{fan2021predicting,chen2019deep,zheng2019utility} and deep neural networks \cite{tallapally2018user,nassar2020novel,shambour2021deep}. However, such prior studies are unable to explicitly capture the collaborative signal in complex contextual semantics across MC ratings.

On the other hand, graph neural networks (GNNs) have been applied to a wide range of recommendation tasks ({\it e.g.}, LightGCN \cite{he2020lightgcn} in single-rating recommendation, SURGE \cite{chang2021sequential} in sequential recommendation,  MBGC \cite{jin2020multi} in multi-behavior recommendation, and SR-GNN \cite{wu2019session} in session recommendation) due to GNN's high expressive capability in learning high-order proximity among users and items, resulting in accurate recommendation results. Surprisingly, albeit their state-of-the-art performance, there has been no prior attempt to design MC recommender systems using GNNs. In this context, even with a number of studies on CF-based MC recommendation \cite{nassar2020novel,tallapally2018user,shambour2021deep,fan2021predicting}, a natural question arising is: ``Is it beneficial to take advantage of GNNs for solving the MC recommendation problem in terms of both {\it effectiveness} and {\it efficiency}?"

To answer this question, we make the first attempt towards designing a {\it lightweight} GNN-aided MC recommender system. Rather than straightforwardly adopting existing GNN-based recommendation methods for MC recommendation, we devise {\it our own} methodology, built upon new design principles and comprehensive empirical findings. To this end, we outline two design challenges that must be addressed when building a new GNN-based MC recommendation method:
\begin{itemize}
\item {\bf Graph construction:} which graph type should be taken into account to explore the collaborative signal in MC ratings;
\item {\bf Criteria preference awareness:} how to maximally grasp the criteria preference of users through graph convolution.
\end{itemize}
\noindent\textbf{(\underline{Idea 1})} In designing recommender systems only using single ratings, it is natural to construct a bipartite graph by establishing edges based on user--item interactions as in \cite{berg2017graph,ying2018graph,seo2022siren,wang2019neural,he2020lightgcn}. On the other hand, in MC recommender systems, rather than using one bipartite graph, the graph construction step accompanies non-straightforward design choices along with MC ratings. To harness the expressive capability of GNNs in learning representations, GNNs should be able to leverage the complex behavioral similarity in the high-order connectivities {\it across} MC ratings. Modeling MC ratings as a multi-graph with multi-relations can be one possible option. However, GNNs designed for such heterogeneous graphs mostly require excessive computational costs or handcrafted meta-paths \cite{wang2019heterogeneous,guo2020survey,lv2021we}, which are undesirable as our design principle since we aim to build a lightweight model with few learnable parameters. As an alternative, we present an {\it MC expansion graph} that transforms MC ratings into an expanded bipartite graph to potentially learn from the collaborative signal in MC ratings. Concretely, in the constructed graph, each item is expanded to different {\it criterion-item} nodes. Figure \ref{fig:intro_1b} illustrates the MC expansion graph constructed by creating some edges, corresponding to high ratings, i.e., the rating scores of 3--5, between a user node and {\it criterion-item} nodes. Our graph construction enables multi-layer GNNs to effectively capture complex contextual semantics existing among MC ratings.

\noindent\textbf{(\underline{Idea 2})} We are interested in designing a GNN model that is capable of making full use of MC rating information based on the constructed MC expansion graph. Meanwhile, users tend to make decisions according to their preferences {\it w.r.t.} one or multiple aspects (criteria) of items \cite{nassar2020novel, shambour2021deep}. For example, in a hotel recommender system, some users may prefer a hotel based on its cleanliness while others may like the same hotel for its price, check-in service, or any other combinations of the distinct attributes of that hotel. In light of this, it is of paramount importance to be aware of the criteria preference of each user when we learn representations through graph convolution. As one of our main contributions, we propose a novel GNN architecture, {\it criteria preference-aware} light graph convolution (\textsf{CPA-LGC}), which is capable of precisely capturing the criteria preference of users as well as the collaborative signal in complex high-order connectivities on the MC expansion graph at a fine-grained level. To reinforce the capability of criteria preference awareness, we newly characterize two embeddings, including {\it user-specific criteria-preference (UCP) embeddings and {\it item-specific criteria (IC) embeddings}}, and incorporate them into the graph convolution model. Then, \textsf{CPA-LGC} predicts the user preference by discovering the final representations of user nodes and criterion-item nodes that accommodate the two newly characterized embeddings.

To validate the effectiveness of \textsf{CPA-LGC}, we comprehensively conduct empirical evaluations using large-scale benchmark datasets ({\it e.g.}, x23.0, x5.5, and x12.8 scale of the datasets used in \cite{tallapally2018user}, \cite{nassar2020novel}, and \cite{shambour2021deep} in terms of the number of overall ratings, respectively). Most importantly, experimental results demonstrate that our method significantly and consistently outperforms the best MC recommendation competitor and the best GNN-based recommendation competitor up to 141.20\% and 58.66\%, respectively, in terms of the precision. 

Our main contributions are summarized as follows:
\begin{itemize}
    \item \textbf{Novel methodology:} We propose an MC recommendation method using a novel GNN architecture, named as \textsf{CPA-LGC}, that deliberately captures 1) the collaborative signal in complex high-order connectivities from constructing our MC expansion graph and 2) the criteria preference of users from accommodating two new embeddings ({\it i.e.}, UCP embeddings and IC embeddings).
    
    \item {\bf Analysis and evaluation:} We validate the rationality and effectiveness of \textsf{CPA-LGC} through extensive experiments on four real-world datasets. We demonstrate (a) the superiority over eleven state-of-the-art recommendation methods by a significant margin, (b) the impact of key hyperparameters, (c) the influence of each component in \textsf{CPA-LGC}, (d) the degree of over-smoothing alleviation, and (e) the computational efficiency with linear complexity in the number of ratings.
\end{itemize}

\section{Problem Definition}
In this section, we formally define the top-$K$ MC recommendation, along with basic notations. Let $u \in \mathcal{U}$ and $i\in\mathcal{I}$ denote a user and an item, respectively, where $\mathcal{U}$ and $\mathcal{I}$ denote the sets of all users and all items, respectively. $\mathcal{N}_u \subset \mathcal{I}$ denotes a set of items interacted by user $u$. Then, the top-$K$ MC recommendation problem is defined as follows:

\noindent\textbf{Definition 1: (Top-$K$ MC recommendation)} Given $u \in \mathcal{U}$ and $i\in\mathcal{I}$, and $C+1$ user--item ratings $\mathcal{R}_0 \times \mathcal{R}_1 \times ... \times \mathcal{R}_C $ including an overall rating $\mathcal{R}_0$, the top-$K$ MC recommendation aims to recommend top-$K$ items that user $u\in\mathcal{U}$ is most likely to prefer among his/her non-interacted items in $\mathcal{I} \setminus \mathcal{N}_u$ {\it w.r.t.} the {\it overall} rating by using all $C+1$ user-item MC ratings.

\section{Proposed Method: CPA-LGC}
\label{section 3}
In this section, we describe our methodology, which includes how to construct an MC expansion graph and how to learn criteria preference awareness alongside our proposed \textsf{CPA-LGC} method. Then, we present the optimization in \textsf{CPA-LGC}. Moreover, we provide the model analysis including the computational complexity of \textsf{CPA-LGC} and the relationship with R-GCN.

\subsection{Graph Construction}
\label{section 3.1}

\begin{figure}[t]
        \centering
        \begin{subfigure}[c]{0.25\columnwidth}
                \includegraphics[width=0.90\columnwidth]{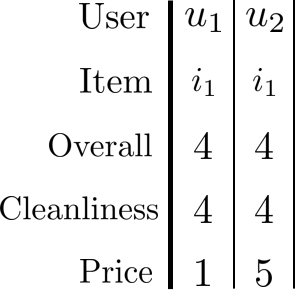}
                \caption{}
                \label{fig:motif_table}
        \end{subfigure}        
        \begin{subfigure}[c]{0.36\columnwidth}
                \includegraphics[width=0.9\columnwidth]{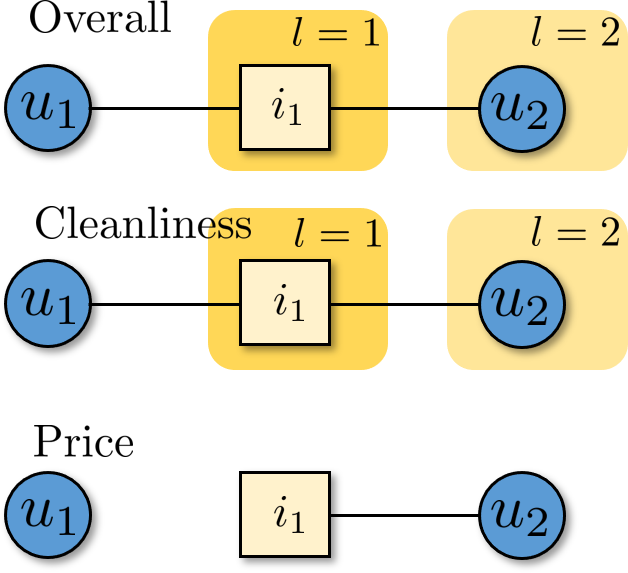}
                \caption{}
                \label{fig:motif_naive}
        \end{subfigure}
        \begin{subfigure}[c]{0.36\columnwidth}
                \includegraphics[width=0.95\columnwidth]{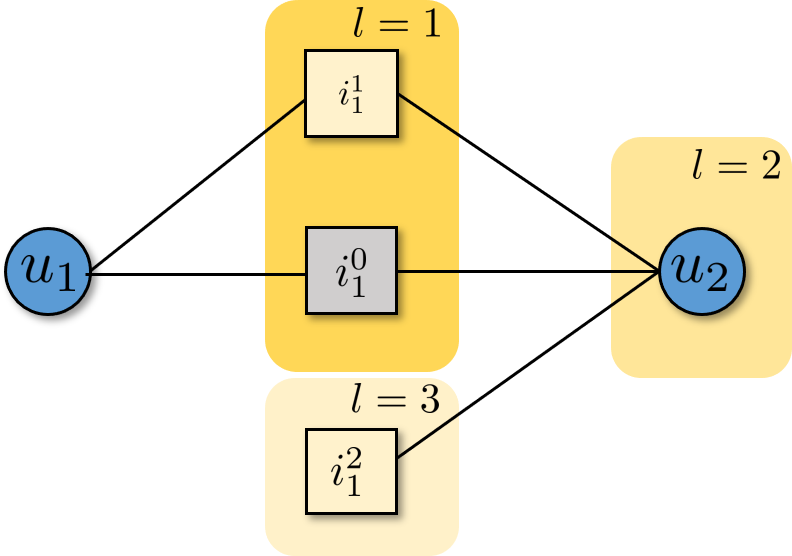}
                \caption{}
                \label{fig:motif_mceg}
        \end{subfigure}
        \caption{An example illustrating (a) a rating instance with three criteria ratings, (b) three graphs, each of which is constructed by ratings per criterion, and (c) our MC expansion graph. In (b) and (c), edges corresponding to the rating scores of 3--5 are created and newly-involved nodes in each GNN layer for target user $u_1$ are marked with different colors.}
        \label{fig1} 
\end{figure}


\noindent{\bf Construction of an MC expansion graph.} A na\"ive graph construction approach using MC rating information would be to construct $C+1$ separate bipartite graphs based on MC ratings including overall ratings. However, in this case, complex contextual semantics existing among multiple user--item interactions cannot be captured via multi-layer GNNs. Figure \ref{fig:motif_table} illustrates a rating instance of hotel $i_1$ with three criteria ratings in the 1--5 rating scale, where two users $u_1$ and $u_2$ reveal a rather \textit{complex} behavioral similarity in that they express the same opinions in terms of the cleanliness, but not in terms of the price. Figure \ref{fig:motif_naive} visualizes three separate bipartite graphs, each of which is constructed by creating some edges, corresponding to high ratings ({\it i.e.}, the rating scores of 3--5), between users and items for each criterion. However, this na\"ive graph construction fails to capture the complex behavioral similarity in the high-order connectivities \textit{across} MC ratings when multi-layer GNNs are employed independently on each graph. For example, the rating information of user $u_2$ {\it w.r.t.} the price cannot be propagated to other graphs where overall and cleanliness ratings are concerned, which limits the high expressive capability of GNNs for acquiring richer representations.

To overcome this inherent limitation, we design a new bipartite graph, namely an \textit{MC expansion graph}, in which each item is expanded to $C+1$ different \textit{criterion-item} nodes, as illustrated in Figure \ref{fig:motif_mceg}. If a user provided a high rating or had a positive interaction for a particular criterion, then an edge between the corresponding user and criterion-item nodes is created. Formally, given a set of criterion-item nodes {\it positively} rated by user $u$ {\it w.r.t.} criterion $c$, denoted as $\mathcal{N}^c_{u}$, the resulting MC expansion graph is denoted as $\mathcal{G} = (\mathcal{V},\mathcal{E})$, where $\mathcal{V}$ and $\mathcal{E}=\{(u,i^c)| u\in\mathcal{U}, i^c \in\mathcal{N}^c_{u},c=0,1,\cdots,C\}$ are the sets of nodes and edges in the graph, respectively. In our setting, the union of all criterion-item node sets positively rated by user $u$ is denoted as $\mathcal{N}_u = \bigcup_{c=0}^C \mathcal{N}_u^{c}$. Note that the graph $\mathcal{G}$ can be modeled as a {\it weighted} graph so that MC rating information is leveraged more precisely. Our MC expansion graph construction enables us to exploit complex {\it high-order proximity} among user nodes and criterion-item nodes with the aid of GNNs. In other words, by feeding the MC expansion graph into GNNs, it is possible to effectively capture the collaborative signal in complex high-order connectivities ({\it i.e.}, complex contextual semantics existing among multiple user--item interactions). Figure \ref{fig:motif_mceg} visualizes our MC expansion graph in which edges corresponding to the rating scores of 3--5 are created. If a 3-layer GNN is applied to the graph, then we are capable of generating user/criterion-item representations that reflect high-order connectivity information. As an example, information of a complex behavioral similarity, such as ``for hotel $i_1$, two users $u_1$ and $u_2$ have the same preference for cleanliness, but different preferences for price", can be incorporated into a vector representation of target user $u_1$ through graph convolution.

\noindent{\bf Over-smoothing effect in the MC expansion graph}. While stacking multiple layers in GNNs is beneficial in capturing the high-order structural information, it may lead to the problem of over-smoothing where node representations converge to a certain value and thus become less distinguishable \cite{li2018deeper, chen2020simple}. This holds particularly strong validity for nodes of a higher degree \cite{chen2020simple,lu2021skipnode}, which exhibit a higher convergence rate {\it w.r.t.} the number of layers in GNNs \cite{chen2020simple}. In the MC expansion graph, the over-smoothing effect may be intensified due to an increased number of item neighbors for each user. To alleviate this problem, we design an additional module to be added to each GNN layer, which will be specified in the following subsection. 

\subsection{Criteria Preference-Aware Architecture}
\begin{figure}[t]
    \centering
    \includegraphics[width=0.49\textwidth]{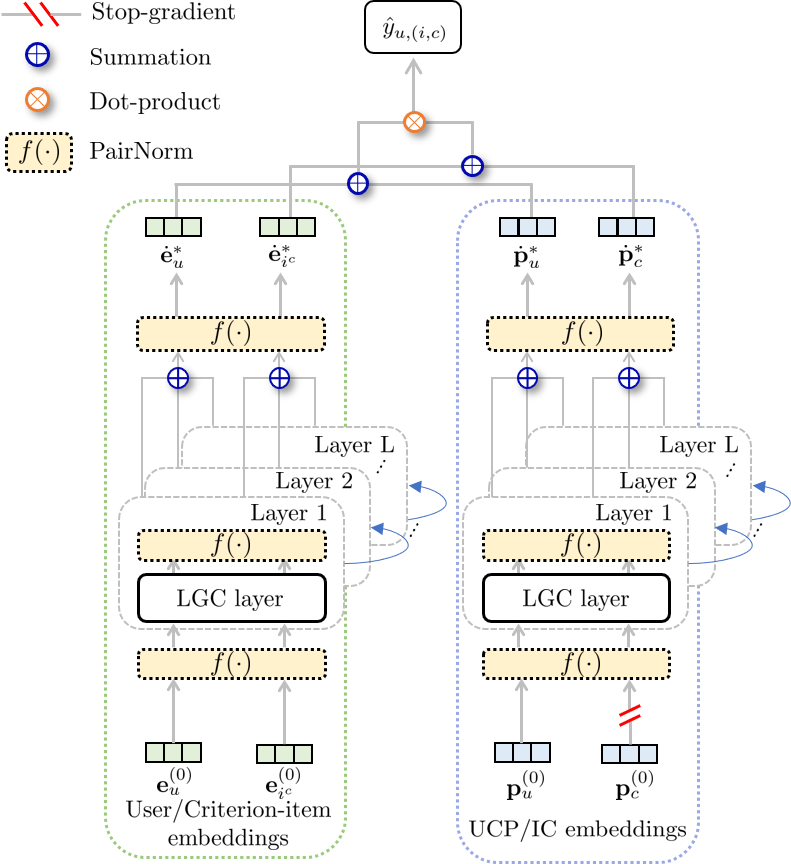}
    \caption{The schematic overview of \textsf{CPA-LGC}.}
    \label{fig:CPA-LGC}
\end{figure}
In this subsection, we elaborate on the four key components of \textsf{CPA-LGC}. The
schematic overview of \textsf{CPA-LGC} is illustrated in Figure \ref{fig:CPA-LGC}.
\subsubsection{Layer-wise Over-smoothing Alleviation.}
As stated in Section \ref{section 3.1}, our MC expansion graph construction may potentially intensify over-smoothing in GNNs. To solve this issue, we employ a {\it layer-wise over-smoothing alleviation} strategy. In our study, we adopt PairNorm \cite{zhao2019pairnorm} as a simple normalization technique such that all pairwise distances between node representations remain unchanged across layers. PairNorm is composed of centering and rescaling steps for each node $v\in\mathcal{V}$ in $\mathcal{G}$ and is expressed as follows:
\begin{equation}
\begin{aligned}
\label{pairnorm}
\mathbf{m}^{(l)}_v = \mathbf{e}^{(l)}_v - \frac{1}{|\mathcal{V}|}\sum_{i=1}^{|\mathcal{V}|}\mathbf{e}^{(l)}_i\\
\dot{\mathbf{e}}^{(l)}_v = s\sqrt{|\mathcal{V}|}\frac{\mathbf{m}^{(l)}_v}{\sqrt{\|\mathbf{E}^{(l)}\|_F^2}}, 
\end{aligned}
\end{equation}
where $\mathbf{e}^{(l)}_v$ is the representation of node $v$ after the $l$-th layer propagation, which will be specified in later; ${\bf m}_v^{(l)}$ is the centered representation of node $v$ after the $l$-th layer propagation; $\dot{\mathbf{e}}^{(l)}_v$ is the output of the PairNorm operation $f(\cdot)$, that is, $\dot{\bf e}_v^{(l)}=f({\bf e}_v^{(l)})$; $\mathbf{E}^{(l)} \in \mathbb{R}^{|\mathcal{V}| \times d}$ is the node representation matrix after the $l$-th layer propagation given the $d$-dimensional latent representation vector of each node; $\|\cdot\|_F$ is the Frobenius norm of a matrix; and $s$ is the scaling hyperparameter that controls the total pairwise squared distance between node representations. In \textsf{CPA-LGC}, we use $f(\cdot)$ after each GNN layer (see Figure \ref{fig:CPA-LGC}), so as to prevent the over-smoothing that may be potentially intensified by the increased degree of each node in the MC expansion graph.
\subsubsection{LGC for User/Criterion-Item Embeddings}

It is known that {\it lightweight} GCN-based models (see \cite{he2020lightgcn, chen2020revisiting, mao2021ultragcn} and references therein), which simplify GCNs by removing feature transformations and/or nonlinear activations, are quite effective in achieving state-of-the-art performance for single rating recommender system. Since utilizing MC ratings inherently accompanies
expensive computational overheads compared to the case of single ratings, it is also vital to adopt a lightweight model for designing MC recommender systems with GNNs while guaranteeing satisfactory performance. In light of this, we build a simple yet effective layer-wise LGC operation in the MC expansion graph, which is formulated as:
\begin{equation}
\begin{aligned}
\label{lgc eq}
\mathbf{e}^{(l)}_u = \sum_{i^c \in \mathcal{N}_u}\frac{w_{u,i^c}}{\sqrt{\sum_{i^c \in \mathcal{N}_u}{w_{u,i^c}}}\sqrt{\sum_{v \in \mathcal{N}_{i^c}}w_{v,i^c}}} \dot{\mathbf{e}}^{(l-1)}_{i^c} \\
\mathbf{e}^{(l)}_{i^c} = \sum_{u \in \mathcal{N}_{i^c}}\frac{w_{u,i^c}}{\sqrt{\sum_{u \in \mathcal{N}_{i^c}}{w_{u,i^c}}}\sqrt{\sum_{j^r \in \mathcal{N}_u}w_{u,j^r}}} \dot{\mathbf{e}}^{(l-1)}_{u}, 
\end{aligned}
\end{equation}
where $\mathbf{e}^{(l)}_u$ and $\mathbf{e}^{(l)}_{i^c}$ indicate the representations of user node $u$ and criterion-item node $i^c$, respectively, after the $l$-th layer propagation (see the left part of Figure \ref{fig:CPA-LGC});\footnote{$\mathbf{e}^{(0)}_{u}$ and $\mathbf{e}^{(0)}_{i^c}$ are the ID embeddings of user node $u$ and criterion-item node $i^c$, respectively.} $w_{u,i^c}$ is the edge weight between node pair $(u,i^c)$; and the denominator in Eq. \eqref{lgc eq} is the symmetric normalization term, which basically follows the design of standard GCN \cite{DBLP:conf/iclr/KipfW17}, to prevent the scale of embeddings from increasing over layers in the weighted graph. Note that $\dot{\mathbf{e}}^{(l)}_u = f(\mathbf{e}^{(l)}_u)$ and $\dot{\mathbf{e}}^{(l)}_{i^c} =f(\mathbf{e}^{(l)}_{i^c})$.


\subsubsection{LGC for UCP/IC Embeddings}
\begin{figure}[t]
    \centering
    \includegraphics[width=0.47\textwidth]{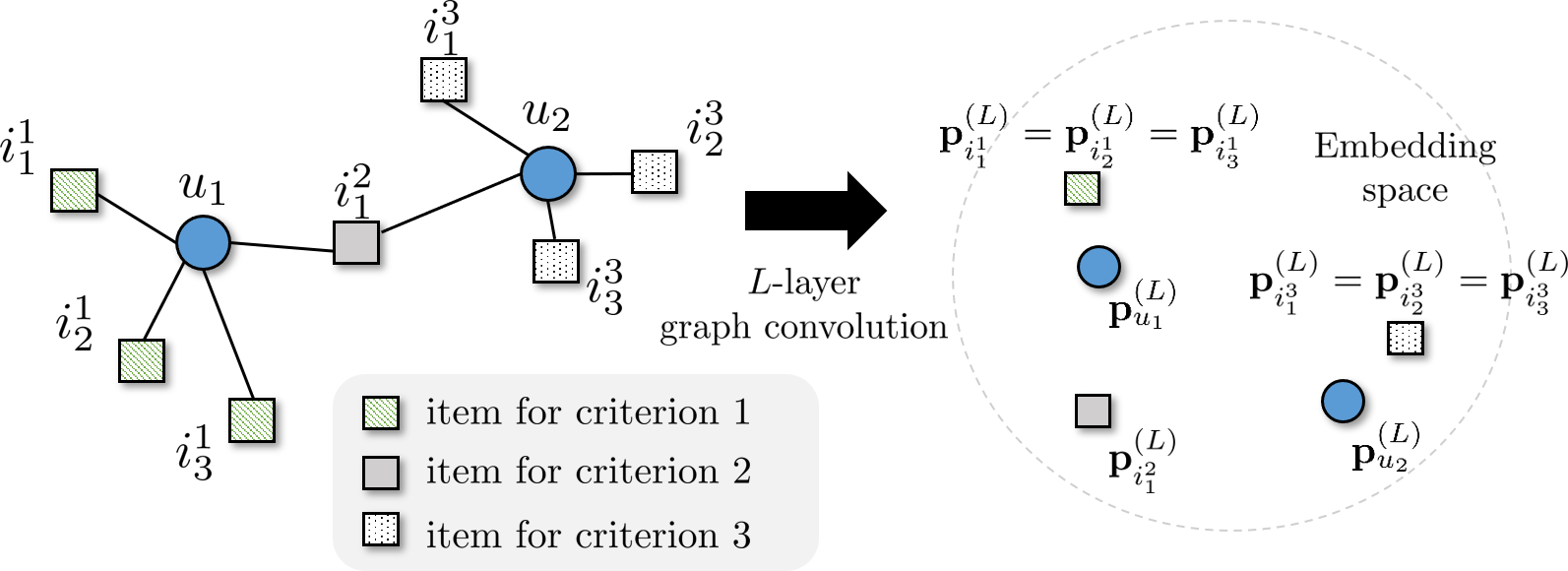}
    \caption{A motivating example describing user criteria preference can be captured via graph convolution. Here, the different IC embeddings are described with different colors and patterns.}
    \label{fig:toy_example}
    \setlength{\textfloatsep}{1\baselineskip plus 0.2\baselineskip minus 0.5\baselineskip}
\end{figure}
As a core component of \textsf{CPA-LGC}, to precisely capture the criteria preference of users, we newly characterize two types of embeddings, including \textit{UCP embeddings} and \textit{IC embeddings}, into our graph convolution model. To generate these newly characterized representations, we formulate the layer-wise LGC operation in the MC expansion graph as follows:
\begin{equation}
\begin{aligned}
\label{criteria lgc eq}
\mathbf{p}^{(l)}_u = \sum_{i^c \in \mathcal{N}_u}\frac{w_{u,i^c}}{\sqrt{\sum_{i^c \in \mathcal{N}_u}{w_{u,i^c}}}\sqrt{\sum_{v \in \mathcal{N}_{i^c}}w_{v,i^c}}} \dot{\mathbf{p}}^{(l-1)}_{i^c} \\
\mathbf{p}^{(l)}_{i^c} = \sum_{u \in \mathcal{N}_{i^c}}\frac{w_{u,i^c}}{\sqrt{\sum_{u \in \mathcal{N}_{i^c}}{w_{u,i^c}}}\sqrt{\sum_{j^r \in \mathcal{N}_u}w_{u,j^r}}} \dot{\mathbf{p}}^{(l-1)}_{u}, 
\end{aligned}
\end{equation}
where $\mathbf{p}^{(l)}_u$ and $\mathbf{p}^{(l)}_{i^c}$ are the UCP embedding of user node $u$ and the IC embedding of criterion-item node $i^c$, respectively, after the $l$-th layer propagation (see the right part of Figure \ref{fig:CPA-LGC}); the denominator in Eq. \eqref{criteria lgc eq} is the symmetric normalization term; and $\dot{\mathbf{p}}^{(l)}_u = f(\mathbf{p}^{(l)}_u)$ and $\dot{\mathbf{p}}^{(l)}_{i^c} =f(\mathbf{p}^{(l)}_{i^c})$. For efficient memory management, we set the initial IC embeddings $\mathbf{p}^{(0)}_{i^c}$ belonging to the same criterion $c$ to be the same, generating $C+1$ different initial embeddings that act as distinct labels without being clustered with each other. As depicted in Figure \ref{fig:CPA-LGC}, we utilize the \textit{stop-gradient} operation in the feed-forward process of the IC embeddings ${\bf p}_{i^c}^{(l)}$ to prevent back-propagation of gradients and unnecessarily excessive  computation. 

Now, let us explain the interplay between two embeddings ${\bf p}_u^{(l)}$ and ${\bf p}_{i^c}^{(l)}$ via graph convolution. Due to the fact that stacking multiple layers in GNNs results in an increased similarity of representations among connected nodes \cite{zhao2019pairnorm}, a user node $u$ connected to a number of different criterion-item nodes belonging to the {\it same} criterion $c$ will have its UCP embedding ${\bf p}_u^{(l)}$ that is co-located to the corresponding IC embeddings in the embedding space. Figure \ref{fig:toy_example} illustrates a motivating example where two users $u_1$ and $u_2$ in the given graph are connected to several criterion-item nodes for criterion 1 and criterion 3, respectively; through $L$-layer graph convolution, the UCP embeddings ${\bf p}_{u_1}^{(L)}$ and ${\bf p}_{u_2}^{(L)}$ are more closely located to the IC embeddings whose related criterion is 1 and 3, respectively.\footnote{In general, IC embeddings ${\bf p}_{i^c}^{(l)}$ are not necessarily the same for different $i^c$'s belonging to the same criterion $c$.} By harnessing the UCP embeddings and IC embeddings as well as user/criterion-item embeddings in the prediction stage, we are capable of achieving higher accuracy of MC recommendation, which will be verified in Section \ref{section 4.2}.

\subsubsection{Layer Combination and Prediction}
\label{section 3.2.3}
The layer combination operation is known to be effective in the sense of capturing different semantics for each layer and alleviating the potential over-smoothing problem \cite{he2020lightgcn, wang2019neural, xu2018representation}. Thus, \textsf{CPA-LGC} leverages the layer combination ({\it i.e.}, layer aggregation) to obtain the combined embeddings, while setting the importance of each layer-wise representation uniformly since such a setting leads to good performance in general \cite{he2020lightgcn}. The combined representations after $L$-layer propagation are expressed as
\begin{equation}
\label{layer comb}
\begin{aligned}
    \mathbf{e}^*_u\! =\! \frac{1}{L}\sum_{l=0}^L{\dot{\mathbf{e}}_u^{(l)}}; \mathbf{e}^*_{i^c} \!=\! \frac{1}{L}\sum_{l=0}^L{\dot{\mathbf{e}}_{i^c}^{(l)}}; \mathbf{p}^*_u \!=\! \frac{1}{L}\sum_{l=0}^L{\dot{\mathbf{p}}_u^{(l)}}; \mathbf{p}^*_{i^c}\! =\! \frac{1}{L}\sum_{l=0}^L{\dot{\mathbf{p}}_{i^c}^{(l)}},
\end{aligned} 
\end{equation}
where $\mathbf{e}^*_u$ and $\mathbf{e}^*_{i^c}$ are the combined embeddings of user $u$ and criterion-item node $i^c$, respectively; $\mathbf{p}^*_u$ and $\mathbf{p}^*_{i^c}$ are the combined UCP embedding of user node $u$ and the combined IC embedding of criterion-item node $i^c$, respectively.

Next, \textsf{CPA-LGC} predicts user $u$'s preference for target criterion-item node $i^c$. To this end, we first form the final representations of a user and a criterion-item node as $\dot{\bf e}_u^*+\dot{\bf p}_u^*$ and $\dot{\bf e}_{i^c}^*+\dot{\bf p}_{i^c}^*$, respectively, where $\dot{\bf e}_u^*=f({\bf e}_u^*)$, $\dot{\bf p}_u^*=f({\bf p}_u^*)$, $\dot{\bf e}_{i^c}^*=f({\bf e}_{i^c}^*)$, and $\dot{\bf p}_{i^c}^*=f({\bf p}_{i^c}^*)$. Then, we compute the matching score $\hat{y}_{u,i^c}$ between the final embedding of user $u$ and the final embedding of criterion-item node $i^c$ via the dot-product as follows:
\begin{equation}
\label{prediction}
    \begin{aligned}
        \hat{y}_{u,i^c} = (\dot{\mathbf{e}}^*_u + \dot{\mathbf{p}}^*_u) \cdot (\dot{\mathbf{e}}^*_{i^c} + \dot{\mathbf{p}}^*_{i^c})^\top.
    \end{aligned}
\end{equation}
It is worth noting that we only take into account $\hat{y}_{u,i^0}$ for prediction since we are interested in predicting the preference of user $u$ for target item $i$ {\it w.r.t.} the {\it overall} rating.

\subsection{Optimization}
In \textsf{CPA-LGC}, we adopt the Bayesian personalized ranking (BPR) loss \cite{rendle2012bpr}, which is widely used for optimizing general recommender systems \cite{wang2019neural, he2018adversarial, he2020lightgcn,he2016vbpr,wang2019kgat}, to learn the trainable parameters of \textsf{CPA-LGC}. The BPR loss is built upon the assumption that the preference for the interacted item of a user is likely to be higher than that for the non-interacted item of a user. Then, our loss function is formulated as follows:
\begin{equation}
\begin{aligned}
\label{bpr_update}
-\sum_{u=1}^{|\mathcal{U}|} \sum_{i^c \in \mathcal{N}_u} \sum_{j^r \notin \mathcal{N}_u} \ln{\sigma\left(\hat{y}_{u,i^c} - \hat{y}_{u,j^r}\right)} + \lambda \|\Theta\|_2^2,
\end{aligned}
\end{equation}
where $\Theta$ is all trainable parameters in the model, and
$\lambda$ is a hyperparameter that controls the $L_2$ regularization strength.

\subsection{Model Analysis}
\subsubsection{Complexity Analysis}
We theoretically analyze the computational complexity
of the \textsf{CPA-LGC} method with given graph $\mathcal{G} = (\mathcal{V},\mathcal{E})$ and $d$-dimensional embeddings.
\begin{theorem}
\label{complexity_thm}
The computational complexity of \textsf{CPA-LGC} is at most linear in $|\mathcal{E}|$.
\end{theorem}

 Note that the complexity of \textsf{CPA-LGC} scales {\it linearly} with the number of user--item interactions ({\it i.e.}, overall ratings) due to the fact that $|\mathcal{E}|$, corresponding to the number of MC ratings, is at most $C+1$ times the number of ratings. We empirically validate the average runtime complexity of \textsf{CPA-LGC} in Appendix \ref{app:scalability}.

\subsubsection{Relationship with R-GCN}
\label{setion 3.3}
We bridge LGC on the MC expansion graph and R-GCN \cite{schlichtkrull2018modeling}, a GNN architecture that was designed for handling multiple relations in {\it multi-graphs}. A multi-graph is denoted as $\tilde{G} = (\tilde{\mathcal{V}}, \tilde{\mathcal{E}}, \tilde{\mathcal{C}})$ with nodes $v, j \in \tilde{\mathcal{V}}$ and relations $(v, c, j) \in \tilde{\mathcal{E}}$, where $c \in \tilde{\mathcal{C}}$ is a relation type. By regarding the MC in ratings as multi-relations, the layer-wise message passing operation in R-GCN from user $u$'s perspective is expressed as
\begin{equation}
\begin{aligned}
\label{R-GCN}
\mathbf{e}_u^{(l)}=\sigma\left(\sum_{c \in \mathcal{C}} \sum_{i \in \mathcal{N}_u^c} \frac{1}{|\mathcal{N}_u^c|} \textbf{W}_c^{(l-1)} \mathbf{e}_i^{(l-1)}+{\bf W}_0^{(l-1)} \mathbf{e}_u^{(l-1)}\right), 
\end{aligned}
\end{equation}
where $\mathbf{e}_u^{(l)}$ and $\mathbf{e}_i^{(l)}$ are the representations of user $u$ and item $i$ after the $l$-th layer propagation; $\sigma$ is the non-linear activation function; $\mathcal{N}_u^c$ denotes the set of neighbors of node $u$ under relation $c \in \tilde{\mathcal{C}}$; and ${\bf W}_0^{(l)}$ and ${\bf W}_c^{(l)}$ are transformation matrices of self-connection and relation $c$ at the $l$-th layer, respectively. Now, we establish a relationship between \textsf{CPA-LGC} and R-GCN by passing through several processes. First, we eliminate the non-linear activation, self-connection, and transformation matrix of relation $c$. Second, we concretize relation $c$ via item embeddings $\mathbf{e}_i^{(l)}$ (\textit{i.e.}, $\mathbf{e}_{i^c}^{(l)}$). Then, Eq. \eqref{R-GCN} can be converted to
\begin{equation}
\label{R-GCN_simple}
\mathbf{e}_u^{(l)}= \sum_{c \in \mathcal{C}} \sum_{i \in \mathcal{N}_u^c} \frac{1}{|\mathcal{N}_u^c|} \mathbf{e}_{i^c}^{(l-1)}, 
\end{equation}
which indicates that the user embedding at the $l$-th layer can be calculated by aggregating information as the weighted sum of all the item embeddings related to a certain criterion at the previous layer. Therefore, LGC on the MC expansion graph in \textsf{CPA-LGC}, can also be viewed as a \textit{simplified} R-GCN with the relation transfer, allowing us to capture the collaborative signal in MC ratings along with far fewer parameters.

\section{Experimental Evaluation}
\label{section 4}
In this section, we systematically conduct extensive experiments to address the key research questions (RQs) outlined below:
\begin{itemize}
    \item \textbf{RQ1:} How much does \textsf{CPA-LGC} improve the top-$K$ recommendation over benchmark MC recommendation methods?
    \item \textbf{RQ2:} How much does \textsf{CPA-LGC} improve the top-$K$ recommendation over benchmark GNN-based recommendation methods?
    \item \textbf{RQ3:} How do key parameters affect the performance of \textsf{CPA-LGC}?
    \item \textbf{RQ4:} How does each component in \textsf{CPA-LGC} contribute to the recommendation accuracy?
    \item \textbf{RQ5:} How does the MC expansion graph alleviate the over-smoothing effect alongside \textsf{CPA-LGC}?
\end{itemize}
\subsection{Experimental Settings}
\label{section 4.1}
\begin{table}[t!]
\small
\caption{Statistics of the four datasets used in our experiments. Here, $\gamma$ denotes the ratio of the number of MC ratings to the number of overall ratings.}
\begin{tabular}{ccccccc}
\hline
\multicolumn{1}{c}{\textbf{Dataset}} & {\begin{tabular}[c]{@{}c@{}} \textbf{\# of} \\ \textbf{users}\end{tabular}} & {\begin{tabular}[c]{@{}c@{}} \textbf{\# of} \\ \textbf{items}\end{tabular}} & {\begin{tabular}[c]{@{}c@{}} \textbf{\# of} \\ \textbf{overall ratings}\end{tabular}} & {\begin{tabular}[c]{@{}c@{}} \textbf{\# of} \\ \textbf{MC ratings}\end{tabular}} & {\begin{tabular}[c]{@{}c@{}}$C$  \end{tabular}}  & {\begin{tabular}[c]{@{}c@{}} \textbf{$\gamma$} \end{tabular}} \\ \hline
\multirow{1}{*}{\begin{tabular}[c]{@{}c@{}} \textbf{TA}\end{tabular}}   & 4,265 & 6,275 &  34,383 &  202,859 & 7 & 5.9 \\\hline
\multirow{1}{*}{\begin{tabular}[c]{@{}c@{}} \textbf{YM}\end{tabular}}   & 1,821 & 1,472 & 46,176 & 175,468 & 4 & 3.8  \\\hline
\multirow{1}{*}{\begin{tabular}[c]{@{}c@{}} \textbf{RB}\end{tabular}} & 4,017 & 3,422 & 159,755 & 607,067 & 4  & 3.8  \\\hline
\multirow{1}{*}{\begin{tabular}[c]{@{}c@{}} \textbf{YP}\end{tabular}}   & 58,971 & 19,820 & 445,724 & 1,408,487 & 3 & 3.1 \\\hline
\end{tabular}

\label{datasettable}

\end{table}

\noindent\textbf{Datasets.} We conduct experiments on four real-world datasets, which are widely used in studies on MC recommendation \cite{zheng2019utility, fan2021predicting, tallapally2018user,nassar2020novel, shambour2021deep, li2020learning}: TripAdvisor (TA), Yahoo!Movie (YM), RateBeer (RB), and Yelp-2022 (YP). Here, Yahoo!Movie was collected by requesting the authors of \cite{jannach2014leveraging}, and the other three datasets are publicly available. It is noteworthy that we use relatively large-scale datasets in comparison to prior studies \cite{zheng2019utility, fan2021predicting, tallapally2018user, shambour2021deep, li2020learning}. To ensure data quality, we use each user/item having at least five interactions. Table \ref{datasettable} summarizes some statistics of the four datasets. We provide further details of the datasets in Appendix \ref{app_dataset desc}.

\noindent\textbf{Competitors.} To comprehensively demonstrate the superiority of \textsf{CPA-LGC}, we present five MC recommendation methods (ExtandedSAE \cite{tallapally2018user}, UBM \cite{zheng2019utility}, DMCF \cite{nassar2020novel}, AEMC \cite{shambour2021deep}, and CFM \cite{fan2021predicting}) and six GNN-based recommendation methods (GC-MC \cite{berg2017graph}, SpectralCF \cite{zheng2018spectral}, NGCF \cite{wang2019neural}, DGCF \cite{wang2020disentangled}, LightGCN \cite{he2020lightgcn}, and $\text{LightGCN}_\text{MC}$).
\begin{table*}
\footnotesize
\caption{Performance comparison among \textsf{CPA-LGC} and benchmark MC recommendation methods for the four benchmark datasets. Here, the best ($X$) and second-best performers ($Y$) are highlighted in bold and underline, respectively. The gain against the second performer is calculated by $\frac{X-Y}{Y}\times 100$ (\%).}
  \begin{tabular}{c l c c c c c c c c}
    \Xhline{1.2pt}
    \multirow{2}{*}{Method} &
    \multirow{2}{*}{Metric} &
      \multicolumn{2}{c}{TA} & 
      \multicolumn{2}{c}{YM} &
      \multicolumn{2}{c}{RB} &
      \multicolumn{2}{c}{YP}\\
      
    & & $K=5$ & $K=10$ & $K=5$ & $K=10$ & $K=5$ & $K=10$  & $K=5$ & $K=10$ \\
    \hline
    \multirow{3}{*}{\rotatebox{0}{ExtendedSAE}} 
      & $Precision@K $ & 0.0012 & 0.0011 & \underline{0.0675} & \underline{0.0480} & 0.0210 & 0.0273 & OOM & OOM \\
      & $Recall@K$ & 0.0031 & 0.0092 & \underline{0.0694} & \underline{0.1000} & 0.0144 & 0.0374 & OOM & OOM \\
      & $NDCG@K$ & 0.0012 & 0.0043 & \underline{0.1072} & \underline{0.1154} & 0.0285 & 0.0435 & OOM & OOM \\
      \hline
    \multirow{3}{*}{\rotatebox{0}{UBM}} 
      & $Precision@K $ & 0.0158 & 0.0122 & 0.0160 & 0.0223 & 0.0250 & 0.0288 & 0.0137 & 0.0125 \\
      & $Recall@K$ & 0.0443 & 0.0533 & 0.0264 & 0.0294 & 0.0174 & 0.0355 & 0.0386 & 0.0713 \\
      & $NDCG@K$ & 0.0351 & 0.0346 & 0.0202 & 0.0245 & 0.0301 & 0.0440 & 0.0248 & 0.0341 \\
      \hline
    \multirow{3}{*}{\rotatebox{0}{DMCF}} 
      & $Precision@K$ & 0.0167 & 0.0137 & 0.0334 & 0.0242 & 0.0816 & \underline{0.0721} & 0.0090 & 0.0075\\
      & $Recall@K$ & 0.0174 & 0.0232 & 0.0333 & 0.0470 & 0.0493 & 0.0887 & 0.0102 & 0.0248 \\
      & $NDCG@K$ & 0.0115 & 0.0243 & 0.0541 & 0.0614 & 0.1104 & 0.1317 & 0.0304 & 0.0408 \\
      \hline
    \multirow{3}{*}{\rotatebox{0}{AEMC}} 
      & $Precision@K $ & 0.0156 & 0.0154 & 0.0358 & 0.0257 & \underline{0.0997} & 0.0807 & 0.0093 & 0.0064 \\
      & $Recall@K$ & 0.0172 & 0.0251 & 0.0398 & 0.0540 & 0.0671 & 0.1090 & 0.0437 & 0.0671 \\
      & $NDCG@K$ & 0.0207 & 0.0241 & 0.0595 & 0.0693 & \underline{0.1534} & 0.1780 & \underline{0.0433} & \underline{0.0544} \\
      \hline
    \multirow{3}{*}{\rotatebox{0}{CFM}} 
      & $Precision@K $ & \underline{0.0220} & \underline{0.0170} & 0.0420 & 0.0375 & 0.0739 & 0.0720 & \underline{0.0180} & \underline{0.0165} \\
      & $Recall@K$ & \underline{0.0615} & \underline{0.0740} & 0.0420 & 0.0613 & \underline{0.1111} & \underline{0.1997} & \underline{0.0482} & \underline{0.0891} \\
      & $NDCG@K$ & \underline{0.0487} & \underline{0.0480} & 0.0392 & 0.0583 & 0.1078 & \underline{0.1391} & 0.0349 & 0.0492 \\
      \Xhline{1.2pt}
    \multirow{3}{*}{\rotatebox{0}{CPA-LGC}} 
      & $Precision@K $ & \textbf{0.0449} & \textbf{0.0273} & \textbf{0.1012} & \textbf{0.0788} & \textbf{0.2177} & \textbf{0.1739} &\textbf{0.0360} & \textbf{0.0276} \\
      & $Recall@K$ & \textbf{0.0901} & \textbf{0.1053} &\textbf{0.1211} & \textbf{0.1725} & \textbf{0.1863} & \textbf{0.2745} & \textbf{0.0859} &\textbf{0.1286} \\
      & $NDCG@K$ & \textbf{0.0830} & \textbf{0.0880} & \textbf{0.1392} & \textbf{0.1532} & \textbf{0.2823} & \textbf{0.2892} & \textbf{0.0713} & \textbf{0.0859} \\
      \Xhline{1.2pt}
    
    \multirow{3}{*}{\rotatebox{0}{Gain}} 
      & $Precision@K $ & +104.09 \% & +60.59 \% & +49.93 \% & +64.17 \% & +136.37 \% & +141.20 \% & +100.00 \% & +67.27 \% \\
      & $Recall@K$  & +46.50 \% & +42.30 \% & +74.50 \% & +72.50 \% & +67.69 \% & +37.46 \% & +78.22 \% & +44.33 \% \\
      & $NDCG@K$  & +70.43 \% & +83.33 \% & +29.85 \% & +32.76 \% & +88.96 \% & +107.91 \% & +64.67 \% & +57.90\% \\
      \Xhline{1.2pt}
\label{table_MC_comparison}
  \end{tabular}
\end{table*}

Specifically, for the five GNN-based recommendation methods except for $\text{LightGCN}_\text{MC}$, we use only overall ratings since those five were originally designed by leveraging single ratings. In our study,  we additionally present a variant of LightGCN, dubbed $\text{LightGCN}_\text{MC}$, which applies LightGCN \cite{he2020lightgcn} to each of $C+1$ bipartite graphs constructed by the MC ratings and then concatenates the output representations of user/item nodes for the final prediction. The schematic overview of $\text{LightGCN}_\text{MC}$ is visualized in Figure \ref{fig:LightGCN-MC} of Appendix \ref{app:lightgcn_mc_detail}.

\noindent\textbf{Evaluation protocols.} We randomly select 70\% of the interactions of each user for the training set, the other 10\% for the validation set, and the remaining 20\% for the test set.
To evaluate the accuracy of top-$K$ MC recommendation, we use the precision,
recall, and normalized discounted cumulative gain (NDCG) as performance metrics, where $K$ is set to 5 and 10 by default. In the inference phase, we view user--item interactions in terms of the {\it overall} rating in the test set as positive and evaluate how well each method can rank the items in the test set higher than all unobserved items. We report the average of values obtained by performing the 10 independent evaluations for each measure.
\begin{table*}
\footnotesize
\caption{Performance comparison among \textsf{CPA-LGC} and benchmark GNN-based recommendation methods for the four benchmark datasets. Here, the best ($X$) and second-best performers ($Y$) are highlighted in bold and underline, respectively. The gain against the second performer is calculated by $\frac{X-Y}{Y}\times 100$ (\%).}
  \begin{tabular}{c l c c c c c c c c}
    \Xhline{1.2pt}
    \multirow{2}{*}{Method} &
    \multirow{2}{*}{Metric} &
      \multicolumn{2}{c}{TA} & 
      \multicolumn{2}{c}{YM} &
      \multicolumn{2}{c}{RB} &
      \multicolumn{2}{c}{YP}\\
      
    & & $K=5$ & $K=10$ & $K=5$ & $K=10$ & $K=5$ & $K=10$  & $K=5$ & $K=10$ \\
    \hline
    \multirow{3}{*}{\rotatebox{0}{GC-MC}} 
      & $Precision@K $ & 0.0060 & 0.0055 & 0.0603 & 0.0508 & 0.1543 & 0.1234 & 0.0208 & 0.0175 \\
      & $Recall@K$ & 0.0157 & 0.0284 & 0.0745 & 0.1246 & 0.1762 & 0.2547 & 0.0533 & 0.0895 \\ 						
      & $NDCG@K$ & 0.0112 & 0.0159 & 0.0820 & 0.0978 & 0.2232 & 0.2354 & 0.0409 & 0.0530 \\
      \hline
    \multirow{3}{*}{\rotatebox{0}{SpectralCF}} 
      & $Precision@K $ & 0.0015 & 0.0016 & 0.0594 & 0.0472 & 0.1655 & 0.1306 & 0.0086 & 0.0081 \\
      & $Recall@K$ & 0.0054 & 0.0111 & 0.0798 & 0.1202 & 0.1646 & 0.2424 & 0.0190 & 0.0351 \\
      & $NDCG@K$ & 0.0037 & 0.0058 & 0.0842 & 0.0963 & 0.2255 & 0.2339 & 0.0148 & 0.0204 \\
      \hline
    \multirow{3}{*}{\rotatebox{0}{NGCF}} 
      & $Precision@K $ & 0.0181 &	0.0119 & 0.0814 & 0.0609 & 0.1730 &	0.1380 & 0.0232 & 0.0188 \\
      & $Recall@K$ & 0.0475& 0.0646 &	0.1010 & 0.1156 & 0.1777 & 0.2648 & 0.0600 & 0.0985 \\
      & $NDCG@K$ & 0.0393 & 0.0454 & 0.1139 & 0.1277 & 0.2372	& 0.2534 & 0.0471 & 0.0600 \\
      \hline
    \multirow{3}{*}{\rotatebox{0}{DGCF}} 
      & $Precision@K $ & 0.0265 & 0.0168 & \underline{0.0809} & 0.0618 & 0.1635 & 0.1300 & 0.0223 & 0.0188 \\
      & $Recall@K$ & 0.0729 & 0.0911 & \underline{0.1020} & 0.1506 & 0.1596 & 0.2411 & 0.0494 & 0.0827 \\
      & $NDCG@K$ & 0.0603 & 0.0670 & \underline{0.1150} & 0.1275 & 0.2202 & 0.2296 & 0.0410 & 0.0519 \\
      \hline
    \multirow{3}{*}{\rotatebox{0}{LightGCN}} 
      & $Precision@K $ & 0.0267 & 0.0177 & 0.0771 & 0.0616 & 0.1732 & 0.1382 & 0.0235 & 0.0192 \\
      & $Recall@K$ & 0.0730 & 0.0929 & 0.0959 & 0.1533 & 0.1778 & 0.2622 & 0.0602 & 0.0987 \\
      & $NDCG@K$ & 0.0607 & 0.0671 & 0.1108 & 0.1278 & 0.2384 & 0.2512 & 0.0475 & 0.0603 \\
      \hline
    \multirow{3}{*}{\rotatebox{0}{$\text{LightGCN}_\text{MC}$}} 
      & $Precision@K $ & \underline{0.0283} & \underline{0.0211} & 0.0795 & \underline{0.0656} & \underline{0.1802} & \underline{0.1423} & \underline{0.0267} & \underline{0.0205} \\
      & $Recall@K$ & \underline{0.0799} & \underline{0.0953} & 0.0977 & \underline{0.1566} & \underline{0.1799} & \underline{0.2651} & \underline{0.0633} & \underline{0.1005} \\
      & $NDCG@K$ & \underline{0.0632} & \underline{0.0699} & 0.1122 & \underline{0.1301} & \underline{0.2423} & \underline{0.2566} & \underline{0.0520} & \underline{0.0625} \\
      \Xhline{1.2pt}
    \multirow{3}{*}{\rotatebox{0}{CPA-LGC}} 
      & $Precision@K $ & \textbf{0.0449} & \textbf{0.0273} & \textbf{0.1012} & \textbf{0.0788} & \textbf{0.2177} & \textbf{0.1739} &\textbf{0.0360} & \textbf{0.0276} \\
      & $Recall@K$ & \textbf{0.0901} & \textbf{0.1053} &\textbf{0.1211} & \textbf{0.1725} & \textbf{0.1863} & \textbf{0.2745} & \textbf{0.0859} &\textbf{0.1286} \\
      & $NDCG@K$ & \textbf{0.0830} & \textbf{0.0880} & \textbf{0.1392} & \textbf{0.1532} & \textbf{0.2823} & \textbf{0.2892} & \textbf{0.0713} & \textbf{0.0859} \\
      \Xhline{1.2pt}
    \multirow{3}{*}{\rotatebox{0}{Gain}} 
  & $Precision@K $ & +58.66\% & +29.38\% & +25.09\% & +20.12\% & +20.81\% & +22.21\% & +34.83\% & +34.63\% \\
  & $Recall@K$  & +12.77\% & +10.49\% & +18.73\% & +10.15\% & +3.04\% & +1.55\% & +35.70\% & +28.22\% \\
  & $NDCG@K$  & +31.33\% & +25.89\% & +21.04\% & +17.76\% & +16.51\% & +12.70\% & +37.12\% & +37.44\% \\

      \Xhline{1.2pt}

  \end{tabular}
\label{table_gnn}
\end{table*}

\noindent\textbf{Implementation details.} Unless otherwise stated, we set the dimensionality of the embedding, $d$, to $64$ for all models as in \cite{wang2019neural, he2020lightgcn}. The model parameters including UCP and IC embeddings in \textsf{CPA-LGC} are initialized with the Xavier method \cite{glorot2010understanding}. We use the Adam optimizer \cite{kingma2015adam}, where the mini-batch size is set to 1024. We use the best hyperparameters of competitors and \textsf{CPA-LGC}
obtained by extensive grid search on the validation set in the
following ranges: $\{1e^{-4},5e^{-4},1e^{-3},5e^{-3},1e^{-2}\}$ for the learning rate; $\{1e^{-5},1e^{-4},1e^{-3},1e^{-2}\}$ for the regularization strength $\lambda$; and $\{1,2,3,4,5\}$ for the number of GNN layers, $L$, in the six GNN-based competitors and \textsf{CPA-LGC}. In consequence, we set the hyperparameters as follows: learning rate$=1e^{-3}$; $\lambda = 1e^{-3}$; and $L = 1$ for YM and $L=3$ for other datasets. In \textsf{CPA-LGC}, to accentuate the importance of overall ratings, the edge weight associated with connections to criterion-item nodes for criterion 0 ({\it i.e.}, $w_{u,i^0}$ in Eqs. \eqref{lgc eq} and \eqref{criteria lgc eq}) is set as $\alpha$ while the edge weight for other existing edges is set as 1. The value of $\alpha$ is searched in range of $\{0.5, 1, 1.5, 2, 2.5\}$, and we set $\alpha = 1.5$ for all the datasets unless otherwise specified. In PairNorm, scaling parameter $s$ in Eq. \eqref{pairnorm} is set to $1$. Additionally, we exclude $f(\cdot)$ for the YP dataset as YP reveals the smallest $\gamma$ (see Table \ref{datasettable}), which represents the ratio of the number of overall ratings to the number of MC ratings, and a smaller value of $\gamma$ would lead to a less over-smoothness degree. We implemented \textsf{CPA-LGC} based on Recbole \cite{zhao2021recbole}, an open-sourced
recommendation library. All experiments are carried out with Intel (R) 12-Core (TM) i7-9700K CPUs @ 3.60 GHz and GPU of NVIDIA GeForce RTX 3080. Our source code is available at https://github.com/jindeok/CPA-LGC-Recbole.

\subsection{Results and Analysis}
\label{section 4.2}
In RQ1--RQ4, we provide experimental results on all datasets. For RQ5, we show here only the results on TA due to space limitations, since the results on other datasets showed similar tendencies to those on TA. We evaluate the performance in terms of the NDCG@10 in RQ3--RQ4. Additionally, we highlight the best and the second-best methods in each column of the following tables in bold and underline, respectively.

\noindent\textbf{RQ1: Comparison with five MC recommendation competitors.}
We validate the superiority of \textsf{CPA-LGC} over five MC recommendation competitors through extensive experiments on the four datasets. Table \ref{table_MC_comparison} shows the results of all MC recommendation competitors and \textsf{CPA-LGC}. Our findings are as follows:
\begin{enumerate}[label=(\roman*)]
    \item Expected but surprisingly, \textsf{CPA-LGC} \textit{significantly} and \textit{consistently} outperform all MC recommendation competitors on all datasets regardless of the metrics. Specifically, on TA, YM, RB, and YP, \textsf{CPA-LGC} outperforms best competitors by large margins by up to 104.09\%, 49.93\%, 136.37\%, and 100.00\% in terms of the Precision@5, respectively;
    \item Unlike two-stage approaches (UBM, DMCF, and AEMC) that predict MC ratings excluding overall ratings and then integrate them to infer overall ratings, CFM is a collective matrix factorization method, which robustly shows better results. It means that jointly predicting the overall rating and other MC ratings via CF can be effective for the top-$K$ MC recommendation;
    \item Deep neural network-based methods (ExtendedSAE, DMCF, and AEMC) show satisfactory performance in some cases. Specifically, ExtendedSAE exhibits superb results among the competitors on YM, which is the smallest dataset in terms of the numbers of users and items. However, ExtendedSAE faces an out-of-memory (OOM) problem in the largest dataset, YP, due to its high input/output dimension calculated as the product of the numbers of users and items in the dataset, causing a significant space complexity;
    \item Since the competitors do not use GNNs, they result in far inferior performance compared to \textsf{CPA-LGC} due to their inability to explicitly reflect the complex high-order connectivity in the embedding learning process, which could lead to suboptimal representations \cite{wang2019neural}.
\end{enumerate}

\noindent\textbf{RQ2: Comparison with six GNN-based recommendation competitors.}
We validate the superiority of \textsf{CPA-LGC} over six GNN-based recommendation competitors. Specifically, since there is no prior work exploring GNN-based MC recommendation, we use single ratings ({\it i.e.}, overall ratings) for the existing five GNN-based recommendation methods (GC-MC, SpectiralCF, NGCF, DGCF, and LightGCN), and we further implement a variant of LightGCN ($\text{LightGCN}_\text{MC}$), which is designed for leveraging the MC ratings. Table \ref{table_gnn} shows the results of all GNN-based recommendation competitors and \textsf{CPA-LGC}. Our findings are as follows: 
\begin{enumerate}[label=(\roman*)]
    \item Most importantly, our \textsf{CPA-LGC} also {\it significantly} and {\it consistently} outperforms other competing GNN-based methods on the four datasets, regardless of the metrics. Specifically, on TA, YM, RB, and YP, \textsf{CPA-LGC} outperforms the best competitors by up to 58.66\%, 25.09\%, 20.81\%, and 34.83\% in terms of the Precision@5, respectively;
    \item Among the five competitors using single ratings, LightGCN mostly performs best (with the only exception on YM in case of $K=5$) as it tends to exhibit state-of-the-art performance in wide fields of recommendation \cite{ran2022pm,choi2021lt,hirakawa2021cross,huang2021mixgcf,wu2022graph};
    \item A substantial improvement of $\text{LightGCN}_\text{MC}$ over LightGCN is observed on all the datasets, which implies that na\"ively incorporating MC rating information into graph convolution models is even beneficial in achieving potential gains;

    \item The accuracies of $\text{LightGCN}_\text{MC}$ are far below those of \textsf{CPA-LGC}. This means that capturing the collaborative signal on a bipartite graph constructed by each of MC ratings in a separate manner does not make full use of MC ratings as long as graph convolution is concerned;

    \item Comparing to the results in Table \ref{table_MC_comparison}, the six GNN-based approaches can still be effective against the five non-GNN methods, while mostly showing robust results over all datasets. This again validates our claim that it is beneficial to take advantage of GNNs for accurate top-$K$ MC recommendations.
\end{enumerate}

The above empirical results demonstrate the effectiveness of our MC expansion graph as well as our \textsf{CPA-LGC} that accommodates two new embeddings ({\it i.e.}, UCP embeddings and IC embeddings) for precisely grasping the criteria preference of each user. 

\definecolor{linecol1}{rgb}{1.0, 0.5, 0.0}
\definecolor{linecol2}{rgb}{0.1, 0.6, 0.0}
\definecolor{linecol3}{rgb}{0.2, 0.4, 0.8}
\definecolor{linecol4}{rgb}{0.2, 0.1, 0.7}
\begin{figure}[t!]
\pgfplotsset{footnotesize,samples=10}
\centering
\begin{tikzpicture}
\begin{axis}[
legend columns=-1,
legend entries={TA, YM, RB, YP},
legend to name=named,
xmax=5,xmin=1,ymin= 0,ymax=0.3,
xlabel=(a) Effect of $L$.,ylabel=NDCG@10, width = 3.3cm, height = 3.8cm,
xtick={1,2,3,4,5},ytick={0, 0.1, 0.2, 0.3}]
    \addplot+[color=linecol1] coordinates{(1,0.073) (2,0.086) (3,0.088) (4,0.088) (5,0.085) };
    \addplot+[color=linecol2] coordinates{(1,0.161) (2,0.158) (3,0.153) (4,0.148) (5,0.142) };
    \addplot+[color=linecol3] coordinates{(1,0.280) (2,0.278) (3,0.2892) (4,0.29) (5,0.288) };
    \addplot+[color=linecol4] coordinates{(1,0.055) (2,0.058) (3,0.063) (4,0.064) (5,0.060)};
\end{axis}
\end{tikzpicture}
\begin{tikzpicture}
\begin{axis}[
xmax=5,xmin=1,ymin= 0,ymax=0.3,
xlabel=(b) Effect of $d$., width = 3.3cm, height = 3.8cm, 
xtick={1, 2, 3, 4, 5},xticklabels = {32, 64, 128, 256, 512},ytick={0, 0.1, 0.2, 0.3}]
    \addplot+[color=linecol1] coordinates{(1,0.083) (2,0.086) (3,0.088) (4,0.088) (5,0.082) };
    \addplot+[color=linecol2] coordinates{(1,0.151) (2,0.153) (3,0.155) (4,0.160) (5,0.158) };
    \addplot+[color=linecol3] coordinates{(1,0.280) (2,0.289) (3,0.295) (4,0.299) (5,0.286) };
    \addplot+[color=linecol4] coordinates{(1,0.063) (2,0.070) (3,0.072) (4,0.074) (5,0.073) };

\end{axis}
\end{tikzpicture}
\begin{tikzpicture}
\begin{axis}[
xmax=2.5,xmin=0.5,ymin= 0,ymax=0.3,
xlabel=(c) Effect of $\alpha$., width = 3.3cm, height = 3.8cm,
xtick={0.5,1,1.5,2,2.5},ytick={0, 0.1, 0.2, 0.3}]
    \addplot+[color=linecol1] coordinates{(0.5,0.083) (1,0.086) (1.5,0.088) (2,0.088) (2.5,0.082) };
    \addplot+[color=linecol2] coordinates{(0.5,0.151) (1,0.151) (1.5,0.161) (2,0.151) (2.5,0.155) };
    \addplot+[color=linecol3] coordinates{(0.5,0.28) (1,0.28) (1.5,0.289) (2,0.278) (2.5,0.266) };
    \addplot+[color=linecol4] coordinates{(0.5,0.06) (1,0.06) (1.5,0.064) (2,0.063) (2.5,0.059) };
\end{axis}
\end{tikzpicture}
\vspace{-0.5\baselineskip}
\ref{named}
\caption{The effect of three hyperparameters on the accuracies of \textsf{CPA-LGC}.}
\label{hyper_plot}
\end{figure}
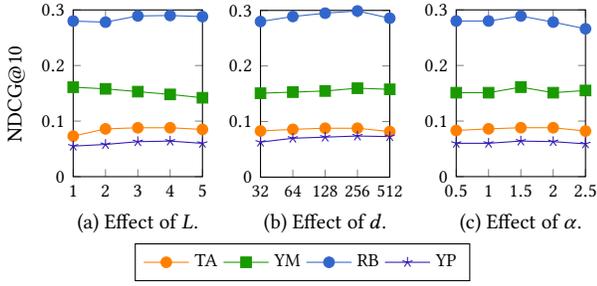

\noindent\textbf{RQ3: Hyperparameter sensitivity analysis.}
In Figure \ref{hyper_plot}, we investigate the impact of three key hyperparameters, including $L$, $d$, and $\alpha$ in \textsf{CPA-LGC}, on the recommendation accuracy.

 {\bf (Effect of $L$)} The number of GNN layers, $L$, decides the degree of exploitation of high-order connectivity among user nodes and criterion-item nodes. From Figure \ref{hyper_plot}a, except for YM, we observe that the recommendation accuracy in terms of the NDCG@10 steadily increases until $L$ reaches 3 and then gradually decreases. This implies that multi-layer LGC is indeed effective in most cases but stacking too many layers may intensify over-smoothing, thereby leading to performance degradation. On the other hand, for the YM dataset, the performance tends to monotonically decrease with $L$, which means that, due to the over-smoothing effect, it is recommended to use only direct neighbors in graph convolution.

{\bf (Effect of $d$)} As shown in Figure \ref{hyper_plot}b, the effect of the dimensionality $d$ of embeddings on the recommendation accuracy is generally observed to be positive  for all the datasets. However, increasing $d$ leads to high computation and overfitting problems \cite{domingos2012few}. The NDCG@10 slightly deteriorates when $d>256$, manifesting the importance of choosing an appropriate $d$ to improve the recommendation accuracy while maintaining the computational efficiency.

 {\bf (Effect of $\alpha$)} The parameter $\alpha$ controls the relative importance of overall ratings compared to MC ratings. From Figure \ref{hyper_plot}c, it is observed that the highest NDCG@10 is achieved at $\alpha=1.5$ regardless of datasets, but further increasing $\alpha$ deteriorates the recommendation accuracy. This implies that overemphasizing overall ratings in LGC may dilute the information acquired from user--item interactions on other criteria and harm the model's performance.

\begin{table}[t!]
    \small
    \centering
    \caption{Performance comparison among \textsf{CPA-LGC} and its three variants in terms of the NDCG@10. Here, the best and second performers are highlighted in bold and underline, respectively.}
    \begin{tabular}{lcccc}
    \toprule
    & \textbf{TA} & \textbf{YM} & \textbf{RB} & \textbf{YP} \\
    \midrule
    \textsf{CPA-LGC} & \textbf{0.088} & \textbf{0.153} & \textbf{0.289} & \underline{0.068}\\
    \textsf{CPA-LGC-MC} & 0.064 & 0.128 & 0.251 & 0.060 \\
    \textsf{CPA-LGC-c} & 0.067 & \underline{0.134} & 0.253 & 0.058 \\
    \textsf{CPA-LGC-f} & \underline{0.070} & 0.131 & \underline{0.259} & \textbf{0.072} \\
    \bottomrule
    \end{tabular}
    \label{ablation_table}
\end{table}

\noindent\textbf{RQ4: Ablation study.}
To analyze the contribution of each component in \textsf{CPA-LGC}, we conduct an ablation study in comparison with three variants depending on which sources are taken into account for designing the \textsf{CPA-LGC} architecture. The performance comparison among the four methods is presented in Table \ref{ablation_table} {\it w.r.t.} the NDCG@10 using four datasets.

\begin{itemize}
\item \textsf{CPA-LGC}: corresponds to the original \textsf{CPA-LGC} method without removing any components;
\item \textsf{CPA-LGC-MC}: uses user embeddings ${\bf e}_u^{(l)}$ and item embeddings ${\bf e}_{i^0}^{(l)}$ for criterion 0 in the LGC operation based on the graph construction only with overall ratings; 
\item \textsf{CPA-LGC-c}: removes UCP embeddings ${\bf p}_u^{(l)}$ and IC embeddings ${\bf p}_{i^c}^{(l)}$ in \textsf{CPA-LGC};
\item \textsf{CPA-LGC-f}: removes the layer-wise over-smoothing alleviation operation $f(\cdot)$.
\end{itemize}
Our observations are as follows: 
\begin{enumerate}[label=(\roman*)]
    \item The original \textsf{CPA-LGC} method always exhibits substantial gains over other variants, which demonstrates that each component plays a crucial role in the success of the proposed method;

    \item The performance gap between \textsf{CPA-LGC} and \textsf{CPA-LGC-MC} tends to be much higher than \textsf{CPA-LGC} and other variants except for YP. This finding indicates that our MC expansion graph is most influential in achieving high recommendation accuracies by precisely capturing the collaborative signal in high-order connectivities between user nodes and criterion-item nodes;

    \item  In comparison with \textsf{CPA-LGC-f}, using $f(\cdot)$ at each layer is shown to yield a positive contribution on the three datasets (TA, YM, and RB) but not on YP. Recall that $\gamma$ in Table \ref{datasettable} is the ratio of the number of MC ratings to 
the number of overall ratings, which signifies the tendency of how much the degree of each node is increased by constructing the MC expansion graph. Since YP has the smallest $\gamma$ ({\it i.e.}, $\gamma=3.1$ from Table \ref{datasettable}) out of all the datasets, over-smoothing may not be severe on YP and thus using $f(\cdot)$ is not beneficial in this case.
\end{enumerate}

\begin{figure}[t]
        \centering      
        \begin{subfigure}[c]{0.46\columnwidth}
                \includegraphics[width=0.99\columnwidth]{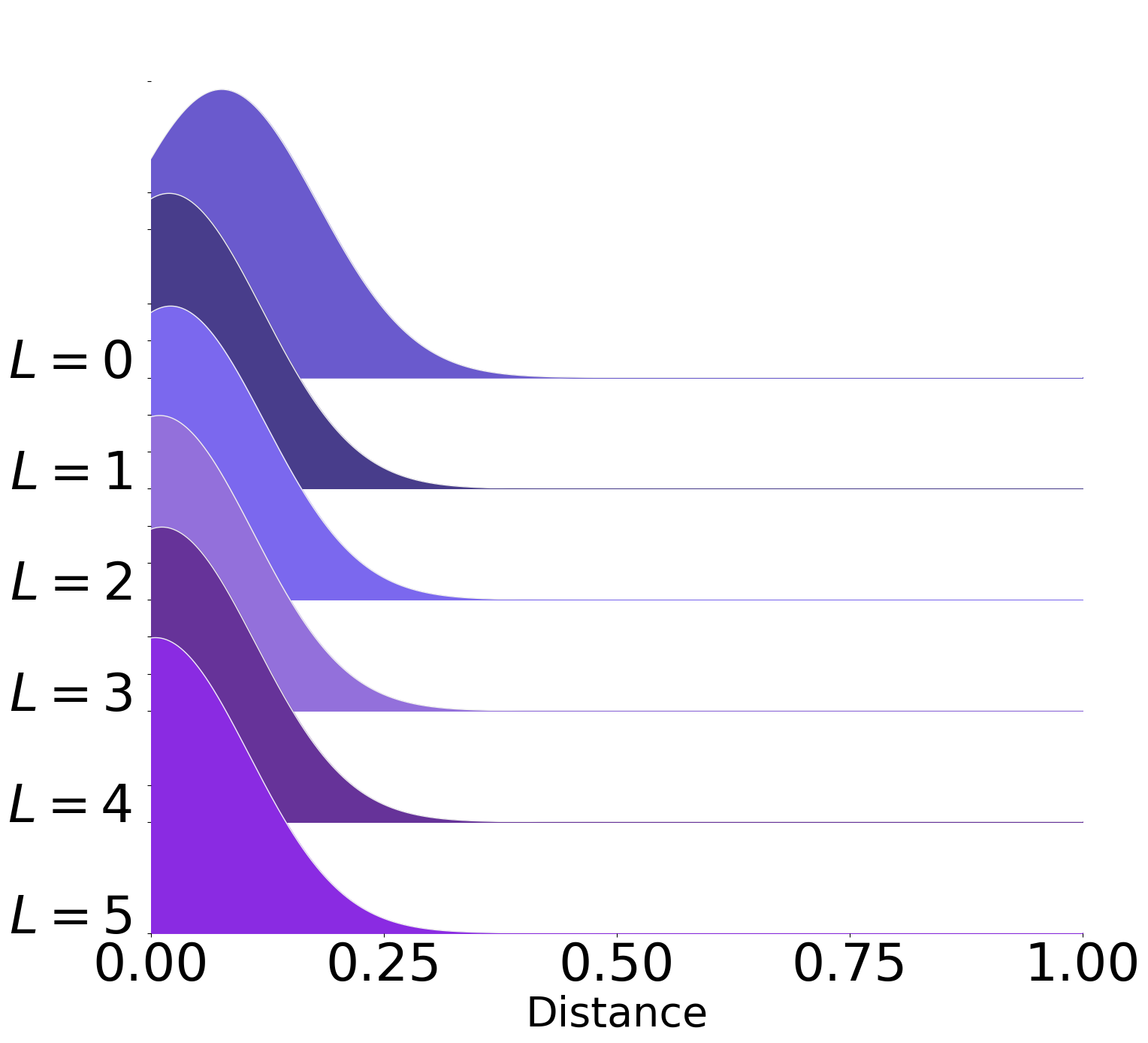}
                \caption{}
                \label{fig:mceg_dist}
        \end{subfigure}
        \begin{subfigure}[c]{0.46\columnwidth}
                \includegraphics[width=0.99\columnwidth]{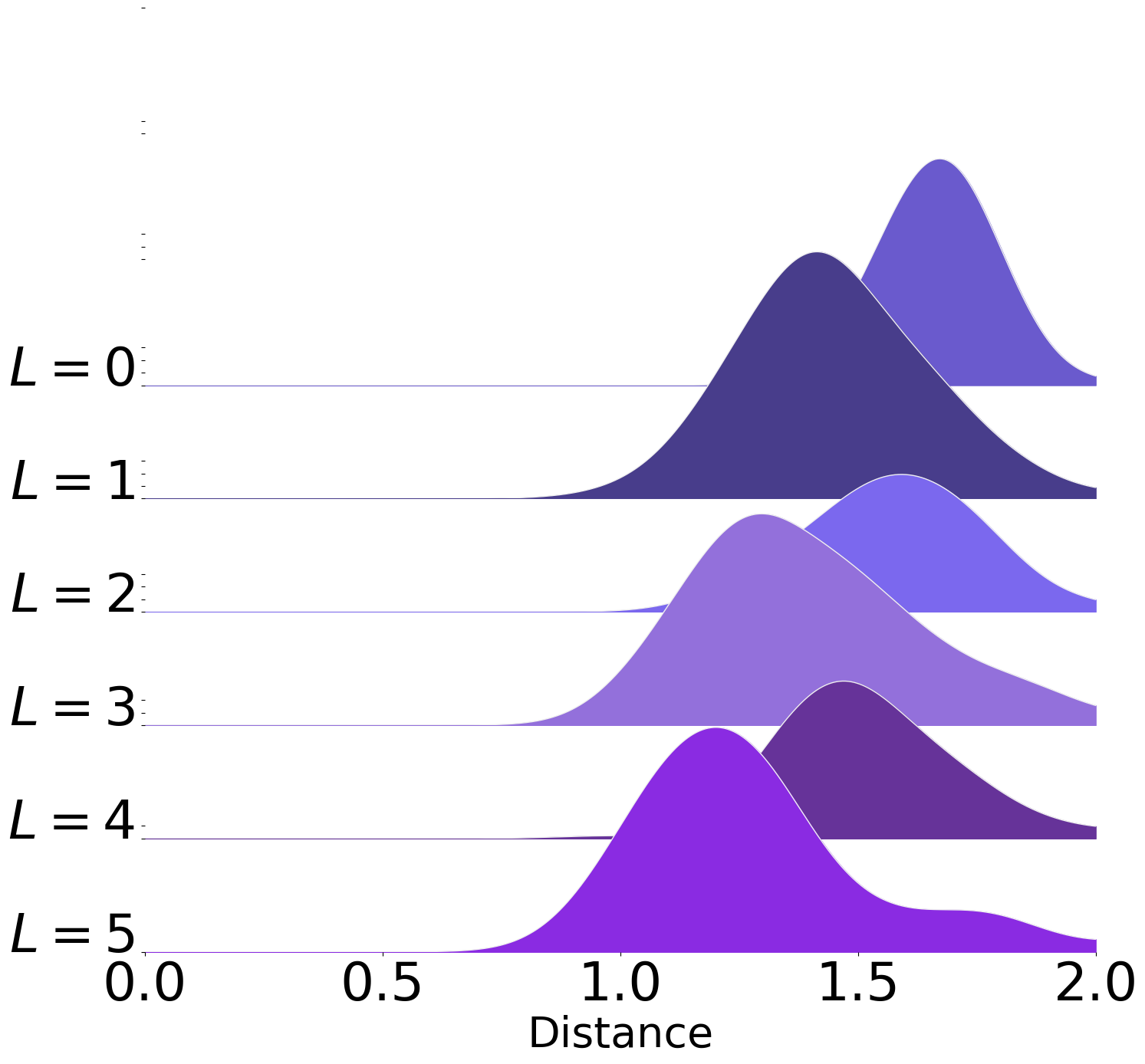}
                \caption{}
                \label{fig:mceg_pr_dist}
        \end{subfigure}
        \caption{Distribution of the Euclidean distances between node representations at each GNN layer on TA for when LGC is performed (a) without $f(\cdot)$ and (b) with $f(\cdot)$. In (b), we only show the distribution of the distances in the range of $[0,2]$ due to space limitations.}
        \label{dist_plot} 
\end{figure}

\noindent\textbf{RQ5: In-depth analysis of the smoothness.}
To validate that over-smoothing can be mitigated by layer-wisely employing the PairNorm operation $f(\cdot)$ in LGC, we analyze the distribution of the Euclidean distances between all node representations at each GNN layer. Figures \ref{fig:mceg_dist} and \ref{fig:mceg_pr_dist} visualize the distributions of such distances when LGC is performed without and with $f(\cdot)$ in the MC expansion graph, respectively, for the TA dataset. One can see that using $f(\cdot)$ at each layer increases the average of the pairwise squared distances between node representations, thereby alleviating potential over-smoothing in the MC expansion graph.
 
 \section{Related Work}
\label{section 5}
In this section, we review some representative methods in two broader fields of research, including 1) MC recommender systems and 2) GNN-based recommender systems.

\noindent\textbf{MC recommender systems.} Efforts have consistently been made to incorporate MC rating information in order to improve the accuracy of recommendations. As an early attempt, a support vector regression-based approach \cite{jannach2012accuracy} was presented to determine the relative importance of the individual criteria ratings. MSVD \cite{li2008improving} was developed by applying a multilinear singular value decomposition technique to capture implicit relations among users, items, and criteria. UBM \cite{zheng2019utility} was proposed by using a utility function in such a way that the user expectations are learned by learning-to-rank methods. CFM \cite{fan2021predicting} was designed by collectively using matrix factorization for MC rating matrices. DTTD \cite{chen2021deep} was developed by incorporating cross-domain knowledge along with side information. Moreover, due to the proliferation of deep learning, there has been a steady push to design DNN-based recommender systems. For example, ExtendedSAE \cite{tallapally2018user} was proposed to capture the relationship between each user's MC and overall ratings using the stacked auto-encoder. LatentMC \cite{li2019latent} was designed with variational auto-encoders to map user reviews into latent vectors, which constitute latent MC ratings. DMCF \cite{nassar2020novel} was developed for predicting MC ratings with a DNN while the predicted ratings are aggregated by another DNN. AEMC \cite{shambour2021deep} was proposed by deep autoencoders, which exploits the non-trivial,
nonlinear, and hidden relations between users with regard to preferences for criteria. However, the aforementioned methods may 1) be unable to explicitly learn the high-order proximity between users and items \cite{jannach2012accuracy,li2008improving,fan2021predicting,tallapally2018user,li2019latent,nassar2020novel,shambour2021deep}, 2) lack scalability \cite{tallapally2018user}, or 3) rely on side information such as user reviews \cite{li2019latent, chen2021deep}. Such limitations of the methods result in unsatisfactory recommendation performance and a lack of robustness to the varying availability of information.

\noindent{\bf GNN-based recommender systems.} GNN-based recommendation has been actively studied accordingly to boost the performance of recommendations. As the first attempt to apply GCN to a recommendation system, GC-MC \cite{berg2017graph} was proposed by taking into account matrix completion for recommender systems from the point of view of link prediction on graphs. PinSage \cite{ying2018graph} was developed by combining the random walk with graph convolution to perform a web-scale recommendation task. SpectralCF \cite{zheng2018spectral} was developed by performing the eigendecomposition on the adjacency matrix of a user--item bipartite graph, so as to discover possible connections between user--item pairs. DGCF \cite{wang2020disentangled} was introduced by separating user intent factors and generating disentangled representations. As one of the follow-up studies, by capturing the high-order collaborative signal existing in user-item interactions, NGCF \cite{wang2019neural} achieved superb performance compared to previous GNN-based approaches. However, extensive ablation studies in LightGCN \cite{he2020lightgcn} convinced that non-linear activation and feature transformation in NGCF are not effective in performing better recommendations; LightGCN has been shown to exhibit state-of-the-art performance in most general recommender systems by removing the two components from the GCN layers in NGCF. Yet, existing GNN-based approaches are limited to single rating recommendation scenarios and do not take into account the MC interactions between users and items.
 
\section{Conclusions and future work}
\label{section 6}
In this paper, we explored an open yet important problem of how to design MC recommender systems with the aid of GNNs. To tackle this challenge, we introduced \textsf{CPA-LGC}, a novel lightweight MC recommendation method that is capable of precisely capturing the criteria preference of users as well as the collaborative signal in MC ratings via LGC. Through extensive experiments on four MC recommendation datasets, we comprehensively demonstrated (a) the superiority of \textsf{CPA-LGC} over eleven benchmark methods, (b) the impact of tuning key hyperparameters in \textsf{CPA-LGC}, (c) the effectiveness of component in \textsf{CPA-LGC}, (d) the degree of over-smoothing alleviation using the PairNorm operation, and (e) the computational efficiency with a linear scaling in $|\mathcal{E}|$.
Potential avenues of our future research include the design of a new GNN architecture that can learn edge weights as trainable parameters along with node representations to automatically learn each user's preference.

\section*{Acknowledgments}
This work was supported by the National Research Foundation of Korea (NRF) grant funded by the Korea government (MSIT) (No. 2021R1A2C3004345, No. RS-2023-00220762).

\bibliographystyle{ACM-Reference-Format}
\balance
\bibliography{sample-base, 0.0.citation_list}

\newpage
\appendix
\newtheorem*{theorem*}{Theorem}
\section{Details of \textsf{CPA-LGC}}
\subsection{Proof of Theorem \ref{complexity_thm}}
\label{app:complexity_thm_proof}
\begin{theorem*}
The computational complexity of \textsf{CPA-LGC} is at most linear in $|\mathcal{E}|$.
\end{theorem*}
\begin{proof}
The feed-forward process of \textsf{CPA-LGC} is composed of two different $L$-layer LGCs including the PairNorm operation for each layer, which is used for generating not only user/criterion-item embeddings but also user UCP and IC embeddings. Since the graph convolution in \textsf{CPA-LGC} is equivalent to LGC in a weighted graph, the computational complexity of the feed-forward process of $L$-layer LGC is $\mathcal{O}(Ld|\mathcal{E}|)$ (refer to \cite{wu2021self} for more details). The $L$-layer PairNorm operation has a complexity of $\mathcal{O}(L|\mathcal{V}|d)$ \cite{zhao2019pairnorm}. Hence, the computational complexity of \textsf{CPA-LGC} is finally given by $\mathcal{O}(Ld(|\mathcal{V}|+|\mathcal{E}|)=\mathcal{O}(Ld|\mathcal{E}|)$, indicating a linear complexity in $|\mathcal{E}|$, which completes the proof of this theorem.
\end{proof}

\subsection{Matrix Form of \textsf{CPA-LGC}}
\label{app:matrix form}
We provide the matrix form of the LGC operation in \textsf{CPA-LGC} to facilitate implementation. Let $\mathbf{A}\in \mathbb{R}^{|\mathcal{V}|\times |\mathcal{V}|}$ be the adjacency matrix of the MC expansion graph, which is a weighted bipartite graph whose entry indicates edge weights. Then, we can calculate the normalized adjacency matrix as $\Tilde{\mathbf{A}} = \mathbf{D}^{-\frac{1}{2}}\mathbf{A}\mathbf{D}^{-\frac{1}{2}}$, where $\mathbf{D} \in \mathbb{R}^{|\mathcal{V}|\times |\mathcal{V}|} $ is the diagonal matrix whose each entry $\mathbf{D}_{ii}$ denotes the summation of nonzero entries in the $i$-th row vector of $\mathbf{A}$. Let denote the two embedding matrices at each GNN layer as $\mathbf{E}^{(l)} \in \mathbb{R}^{|\mathcal{V}|\times d}$ and $\mathbf{P}^{(l)} \in \mathbb{R}^{|\mathcal{V}|\times d}$. More precisely, $\mathbf{E}^{(l)}_{0 \cdots |\mathcal{U}|-1:}$ and $\mathbf{E}^{(l)}_{|\mathcal{U}| \cdots |\mathcal{V}|:}$ are the user embeddings and the criterion-item embeddings, respectively; and $\mathbf{P}^{(l)}_{0 \cdots |\mathcal{U}|:}$ and $\mathbf{P}^{(l)}_{|\mathcal{U}| \cdots |\mathcal{V}|:}$ are the UCP embeddings and the IC embeddings for each criterion-item node, respectively. Then, Eq. \eqref{lgc eq} and Eq. \eqref{criteria lgc eq} can be reformulated as follows:
\begin{equation}
\begin{aligned}
    \dot{\mathbf{E}}^{(l)} = f\left(\tilde{\mathbf{A}}\mathbf{E}^{(l-1)}\right)\\
    \dot{\mathbf{P}}^{(l)} = f\left(\tilde{\mathbf{A}}\mathbf{P}^{(l-1)}\right),
\end{aligned}
\end{equation}
 where $f(\cdot)$ is the layer-wise over-smoothing alleviation. Likewise, the layer-wise combination in Eq. \eqref{layer comb} can also be formulated as a matrix form:
\begin{equation}
\label{comb_matrix}
\begin{aligned}
   \dot{\mathbf{E}}^* = f\left(\frac{1}{L}\sum_{l=1}^{L}{\dot{\mathbf{E}}^{(l)}}\right)\\
\dot{\mathbf{P}}^* = f\left(\frac{1}{L}\sum_{l=1}^{L}{\dot{\mathbf{P}}^{(l)}}\right).
\end{aligned}
\end{equation}
Finally, the prediction can be calculated based on the summation of the two matrices in Eq. \eqref{comb_matrix} ({\it i.e.}, $\dot{\mathbf{E}}^* + \dot{\mathbf{P}}^*$), where a prediction $\hat{y}_{u,i^0}$ is calculated by the dot-product of the $u$-th row vector and the $i$-th row vector of $\dot{\mathbf{E}}^* + \dot{\mathbf{P}}^*$.

\subsection{Pseudocode of \textsf{CPA-LGC}}
\label{app:pseudo-code}

We summarize the training process of \textsf{CPA-LGC} in
Algorithm \ref{cpalgc_pseudo}.

\begin{algorithm}
\caption{\textsf{CPA-LGC}}
\label{cpalgc_pseudo}
\begin{algorithmic}[1]
  \renewcommand{\algorithmicrequire}{\textbf{Input:}}
  \renewcommand{\algorithmicensure}{\textbf{Output:}}
\REQUIRE MC ratings $\mathcal{R}_0 \times \mathcal{R}_1 \times ... \times \mathcal{R}_C$, set of users $\mathcal{I}$, set of items $\mathcal{U}$, maximum epoch $ep$, number of layers $L$
\ENSURE Updated $\mathbf{e}^{(0)}_u$, $\mathbf{e}^{(0)}_{i^c}$, and $\mathbf{p}^{(0)}_u$ for $u \in \mathcal{I}$ and $i^c \in \mathcal{N}_u$
\\
/* MC expansion graph construction */
\STATE Construct $\mathcal{G}$ based on $\mathcal{R}_0 \times \mathcal{R}_1 \times ... \times \mathcal{R}_C$
\\
/*\textsf{CPA-LGC}*/
\STATE Initialize $\mathbf{e}^{(0)}_u$, $\mathbf{p}^{(0)}_u$ for $u \in \mathcal{I}$, $\mathbf{e}^{(0)}_{i^c}$ for $i^c \in \mathcal{N}_u$, and $\mathbf{p}^{(0)}_{c}$ for $c \in \{0,1,2,\cdots,C\}$  
\FOR {$epoch \leftarrow 1$ to $ep$}
\FOR {$l\leftarrow1$ to $L$}
\STATE Obtain $\dot{\mathbf{e}}^{(l-1)}_u$, $\dot{\mathbf{e}}^{(l-1)}_{i^c}$, $\dot{\mathbf{p}}^{(l-1)}_u$, and $\dot{\mathbf{p}}^{(l-1)}_c$ via Eq. \eqref{pairnorm}
\STATE Obtain $\mathbf{e}^{(l)}_u$,$\mathbf{e}^{(l)}_{i^c}$,$\mathbf{p}^{(l)}_u$, and $\mathbf{p}^{(l)}_c$ via Eqs. \eqref{lgc eq} and \eqref{criteria lgc eq}
\ENDFOR
\STATE Obtain $\mathbf{e}^{*}_u$, $\mathbf{e}^{*}_{i^c}$,$\mathbf{p}^{*}_u$, and $\mathbf{p}^{*}_c$ via Eq. \eqref{layer comb}
\STATE Calculate $\hat{y}_{u,i^c}$ via Eq. \eqref{prediction}
\STATE Update $\mathbf{e}^{(0)}_u$, $\mathbf{e}^{(0)}_{i^c}$, and $\mathbf{p}^{(0)}_u$ via Eq. \eqref{bpr_update}
\ENDFOR

\RETURN Updated $\mathbf{e}^{(0)}_u$, $\mathbf{e}^{(0)}_{i^c}$, and $\mathbf{p}^{(0)}_u$ for $u \in \mathcal{I}$ and $i^c \in \mathcal{N}_u$
\end{algorithmic}
\end{algorithm}
\begin{figure}[t]
        \centering
        \begin{subfigure}[c]{0.4\columnwidth}
                \includegraphics[width=0.95\columnwidth]{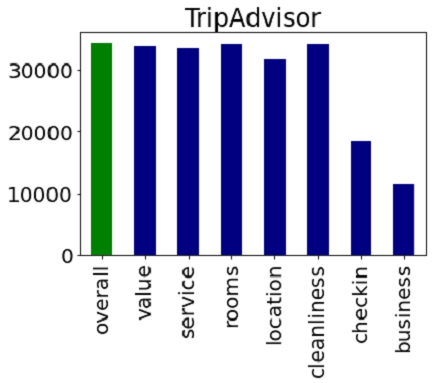}
                \caption{TA}
                \label{fig:ta_rating}
        \end{subfigure}        
        \begin{subfigure}[c]{0.4\columnwidth}
                \includegraphics[width=0.95\columnwidth]{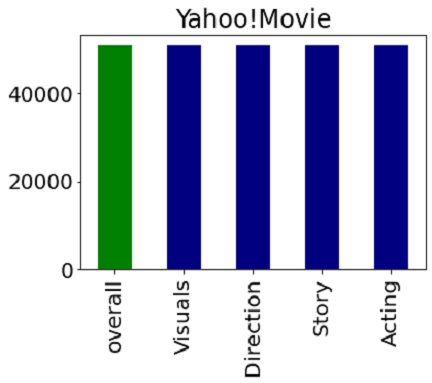}
                \caption{YM}
                \label{fig:ym_rating}
        \end{subfigure}
        \begin{subfigure}[c]{0.4\columnwidth}
                \includegraphics[width=0.95\columnwidth]{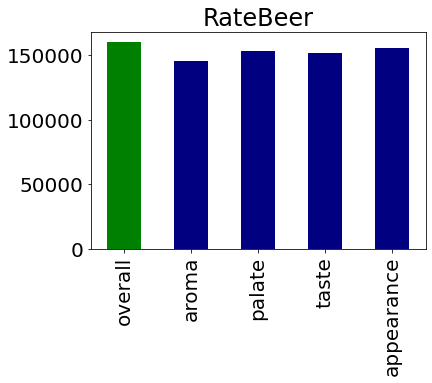}
                \caption{RB}
                \label{fig:rb_rating}
        \end{subfigure}
        \begin{subfigure}[c]{0.4\columnwidth}
                \includegraphics[width=0.95\columnwidth]{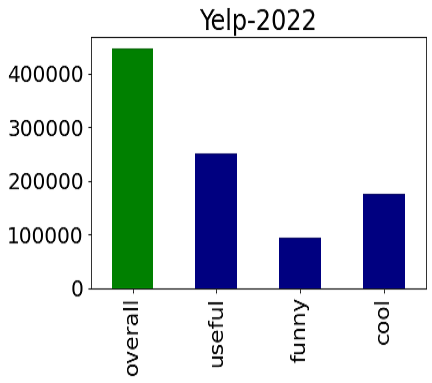}
                \caption{YP}
                \label{fig:yp_rating}
        \end{subfigure}
        \caption{Number of ratings for each criterion in the four datasets.}
        \label{data_ratings_plot} 
\end{figure}



\section{Details of experimental settings}
\subsection{Dataset description}
\label{app_dataset desc}

We describe the details of the datasets used in our experiments. The number of ratings for each criterion shown in Figure \ref{data_ratings_plot}.

\noindent\textbf{TripAdvisor (TA)}: The TA dataset, released by \cite{wang2011latent}, comprises hotel rating information, including an overall rating as well as ratings for seven comprehensive criteria: \textit{business}, \textit{check-in quality}, \textit{cleanliness}, \textit{location}, \textit{rooms}, \textit{service}, and \textit{value}. The ratings are on a scale of 1 to 5 for all criteria.

\noindent\textbf{Yahoo!Movie (YM)}: The TM dataset, first introduced by \cite{jannach2014leveraging}, comprises movie rating information, including an overall rating as well as ratings for four specific criteria: \textit{story}, \textit{acting}, \textit{direction}, and \textit{visuals}. The ratings are on a scale of 1 to 5 for all criteria.

\noindent\textbf{RateBeer (RB)}: The RB dataset, released by \cite{mcauley2012learning, mcauley2013amateurs}, comprises beer rating information, including an overall rating as well as ratings for four specific criteria: \textit{appearance}, \textit{aroma}, \textit{taste}, and \textit{palate}. The ratings range from 1 to 5 (\textit{appearance} and \textit{palate}), 1 to 10 (\textit{aroma} and \textit{taste}), and 1 to 20 (overall).

\noindent\textbf{Yelp-2022 (YP)}: The YP dataset is the most recent version of the Yelp dataset (https://www.yelp.com/dataset). It contains information on MC interactions, including the number of votes for three criteria \textit{cool}, \textit{funny}, and \textit{useful}, in addition to an overall rating on a scale of 1 to 5.

 For the graph construction of each dataset, edges were included in the corresponding graph if the ratings surpass the median value of the rating range, with the exception of Yelp-2022. For Yelp-2022, edges were included if there is at least one vote present for each criterion. 

\subsection{Details of $\text{LightGCN}_\text{MC}$}
\label{app:lightgcn_mc_detail}
\setcounter{figure}{9}
\begin{figure}[t]
    \centering
    \includegraphics[width=0.47\textwidth]{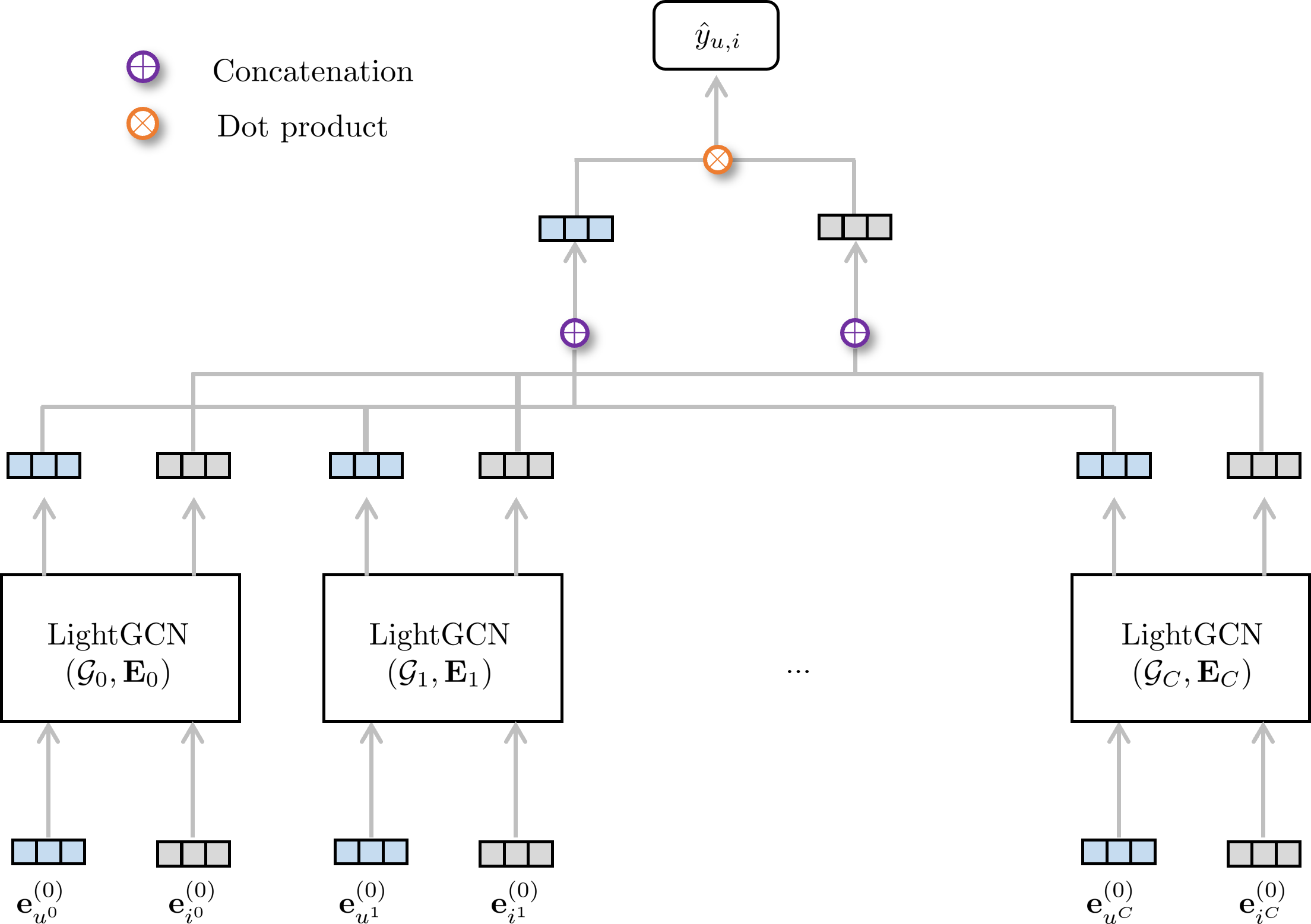}
    \caption{The schematic overview of $\text{LightGCN}_\text{MC}$.}
    \label{fig:LightGCN-MC}
\end{figure}
Figure \ref{fig:LightGCN-MC} describes the schematic overview of $\text{LightGCN}_\text{MC}$. Specifically, given a bipartite graph $\mathcal{G}_c$ and its node embedding matrix $\mathbf{E}_c$ for criterion $c \in \{0,1,\cdots,C\}$, $\text{LightGCN}_\text{MC}$ separately performs LightGCN on $C+1$ different graphs, denoted as $\mathcal{G}_0, \mathcal{G}_1, \cdots, \mathcal{G}_C$, and $\mathbf{E}_0, \mathbf{E}_1, \cdots, \mathbf{E}_C$. Then, the output representations of user/item nodes for each criterion are aggregated by concatenation for the final prediction.

\subsection{Quantitative Analysis of Over-smoothing}
\label{app:distance}

Table \ref{apsd_table} shows the average pairwise Euclidean distance of node representations at each GNN layer on the TA dataset  when \textsf{CPA-LGC} is performed with and without the over-smoothing alleviation operation $f(\cdot)$. It is seen that the average pairwise distance in the case of using $f(\cdot)$ is consistently larger than its counterpart ({\it i.e.}, the case of not using $f(\cdot)$) over all layers.
\begin{table}[]
\small
\caption{The average pairwise Euclidean distance of node representations at each layer on TA when \textsf{CPA-LGC} is performed with and without the over-smoothing alleviation operation $f(\cdot)$.}
\begin{tabular}{cllllll}
\hline 
\multicolumn{1}{c}{Layer} & \multicolumn{1}{c}{0} & \multicolumn{1}{c}{1} & \multicolumn{1}{c}{2} & \multicolumn{1}{c}{3} & \multicolumn{1}{c}{4} & \multicolumn{1}{c}{5} \\ \hline
\begin{tabular}[c]{@{}c@{}}\textsf{CPA-LGC} w/o $f(\cdot)$\end{tabular}      & 0.092                 & 0.043                 & 0.031                 & 0.027                 & 0.019                 & 0.014                 \\ \hline
\begin{tabular}[c]{@{}c@{}}\textsf{CPA-LGC} w/ $f(\cdot)$\end{tabular} & 2.061                 & 1.414                 & 2.546                 & 1.447                 & 2.420                  & 1.369                 \\ \hline
\end{tabular}

\label{apsd_table}
\end{table}
\begin{figure}
\setcounter{figure}{7}
\begin{tikzpicture}
    \begin{axis}[        xlabel=$|\mathcal{E}|$,        ylabel= Time (s), grid=major, grid style={dashed},  width = 7cm, height = 5cm , legend style={at={(0.55,0.2)},anchor=west}]
        \addplot [color=blue, mark=square, legend = \textsf{CPA-LGC}]
            coordinates {
                (10000,0.03) (50000,0.12) (100000,0.24) (500000,1.32) (1000000,2.73) (2000000, 4.89) (3000000,7.52)
            };
        \addplot [ dashed, color=red,            mark=triangle, legend = $O(|\mathcal{E}|)$ ]
            coordinates {
                (10000,0.03) (50000,0.12) (100000,0.24) (500000,1.2) (1000000,2.4) (2000000, 4.8) (3000000,7.2)
            };
    \addlegendentry{\textsf{CPA-LGC}}
    \addlegendentry{$O(|\mathcal{E}|)$}
    \end{axis}    
\end{tikzpicture}
\caption{Compuational efficiency of \textsf{CPA-LGC}. The execution time is measured on one epoch for each experiment. Here, a red dashed line indicates a linear scaling in $|\mathcal{E}|$ obtained from Theorem \ref{complexity_thm}.}
\label{scalabilityplot}
\end{figure}

\subsection{Computational Efficiency of \textsf{CPA-LGC}}
\label{app:scalability}

Figure \ref{scalabilityplot} shows the execution time for one epoch training on \textsf{CPA-LGC}. It is observed that the execution time is almost linear with the number of edges, $|\mathcal{E}|$, in the MC expansion graph. Thus, our empirical result validates the theoretical analysis in Theorem \ref{complexity_thm}.

\begin{figure}
\begin{tikzpicture}
\begin{axis}[    ybar,    width=6.7cm, height=3.5cm,   ymin=0.05,    ymax=0.09,    xtick={1,2,3,4,5,6,7,8},      ylabel={NDCG@10},    xlabel={\# of criteria},    bar width=0.15cm,    grid=major,    grid style={dashed,gray!30},]
\addplot coordinates {
    (1, 0.064)
    (2, 0.066)
    (3, 0.071)
    (4, 0.074)
    (5, 0.078)
    (6, 0.085)
    (7, 0.087)
    (8, 0.088)
};

\end{axis}
\end{tikzpicture}
\caption{Performance comparison of \textsf{CPA-LGC} according to the different number of criteria used in TA. Here, using 1 criterion means that only overall ratings are used as a single rating.}
\label{criteria_ta_exp}
\end{figure}

\subsection{Effect of the Number of Criteria}
To investigate the effect of the number of criteria on the performance, we perform an additional ablation study by varying the number of criteria. We use the TA dataset for the experiment. As shown in Figure \ref{criteria_ta_exp}, the recommendation accuracy tends to monotonically increase with the number of criteria. The experimental result demonstrates that MC ratings indeed contain the informative collaborative signal as long as GNNs are concerned, thus resulting in accurate recommendations.









\end{document}